\documentclass[letterpaper]{article} %
\usepackage{aaai2026}  %
\usepackage{times}  %
\usepackage{helvet}  %
\usepackage{courier}  %
\usepackage[hyphens]{url}  %
\usepackage{graphicx} %
\usepackage{natbib}  %
\usepackage{caption} %
\usepackage{algorithm}
\usepackage{paralist}
\usepackage{newfloat}
\usepackage{listings}
\DeclareCaptionStyle{ruled}{labelfont=normalfont,labelsep=colon,strut=off} %
\floatstyle{ruled}
\newfloat{listing}{tb}{lst}{}
\floatname{listing}{Listing}
\title{Universal Safety Controllers with Learned Prophecies\footnote{Authors are ordered alphabetically.}}
\author {
    Bernd Finkbeiner\textsuperscript{\rm 1},
    Niklas Metzger\textsuperscript{\rm 1},
    Satya Prakash Nayak\textsuperscript{\rm 2},
    Anne-Kathrin Schmuck\textsuperscript{\rm 2}
}
\affiliations {
    \textsuperscript{\rm 1}CISPA Helmholtz Center for Information Security, Saarbr\"ucken, Germany\\
    \textsuperscript{\rm 2}Max Planck Institute for Software Systems, Kaiserslautern, Germany\\
    \{finkbeiner, niklas.metzger\}@cispa.de, \{sanayak, akschmuck\}@mpi-sws.org
}

\usepackage{ltl}
\usepackage{mathtools}
\usepackage{xspace}
\usepackage{algorithm}
\usepackage[noend]{algpseudocode}
\usepackage{enumitem}
\usepackage{amssymb}
\usepackage{amsmath}
\usepackage{amsthm}
\usepackage{graphicx}
\usepackage{amsfonts} %
\usepackage{stmaryrd} %
\usepackage{xcolor}
\usepackage{booktabs}
\usepackage{stmaryrd}
\usepackage{thm-restate}

\definecolor{linkC}{HTML}{710580}

\usepackage{tablefootnote}
\usepackage{makecell}
\usepackage{todonotes}
\usepackage{comment}
\usepackage{subcaption}
\usepackage{xcolor}
\usepackage{listings}
\usepackage{relsize}
\usepackage{thm-restate}
\usepackage{bussproofs}
\usepackage{booktabs}
\usepackage{pifont}%
\usepackage{multirow}
\usepackage[capitalise]{cleveref}
\crefname{theorem}{Thm.}{Theorems}
\Crefname{theorem}{Thm.}{Theorems}
\crefname{example}{Ex.}{Examples}
\Crefname{example}{Ex.}{Examples}
\crefname{equation}{}{}
\Crefname{equation}{}{}
\crefname{figure}{Fig.}{Figs.}
\Crefname{figure}{Fig.}{Figs.}
\crefname{algorithm}{Algo.}{Algos.}
\Crefname{algorithm}{Algo.}{Algos.}
\crefname{definition}{Def.}{Defs.}
\Crefname{definition}{Def.}{Defs.}
\crefname{corollary}{Cor.}{Cors.}
\Crefname{corollary}{Cor.}{Cors.}
\crefname{section}{Sec.}{Secs.}
\Crefname{section}{Sec.}{Secs.}

\newcommand\donotshow[1]{}

\usetikzlibrary{arrows}
\usepackage{csquotes}

\newtheorem{theorem}{Theorem}
\newtheorem{corollary}{Corollary}
\newtheorem{definition}{Definition}

\newtheorem{example}{Example}
\newtheorem{lemma}{Lemma}

\makeatletter
\newcommand{\superimpose}[2]{{%
		\ooalign{%
			\hfil$\m@th#1\@firstoftwo#2$\hfil\cr
			\hfil$\m@th#1\@secondoftwo#2$\hfil\cr
		}%
}}
\makeatother

\newcommand{\nat}{\mathbb{N}}

\makeatletter
\DeclareFontFamily{OMX}{MnSymbolE}{}
\DeclareSymbolFont{MnLargeSymbols}{OMX}{MnSymbolE}{m}{n}
\SetSymbolFont{MnLargeSymbols}{bold}{OMX}{MnSymbolE}{b}{n}
\DeclareFontShape{OMX}{MnSymbolE}{m}{n}{
    <-6>  MnSymbolE5
   <6-7>  MnSymbolE6
   <7-8>  MnSymbolE7
   <8-9>  MnSymbolE8
   <9-10> MnSymbolE9
  <10-12> MnSymbolE10
  <12->   MnSymbolE12
}{}
\DeclareFontShape{OMX}{MnSymbolE}{b}{n}{
    <-6>  MnSymbolE-Bold5
   <6-7>  MnSymbolE-Bold6
   <7-8>  MnSymbolE-Bold7
   <8-9>  MnSymbolE-Bold8
   <9-10> MnSymbolE-Bold9
  <10-12> MnSymbolE-Bold10
  <12->   MnSymbolE-Bold12
}{}

\let\llangle\@undefined
\let\rrangle\@undefined
\DeclareMathDelimiter{\llangle}{\mathopen}%
                     {MnLargeSymbols}{'164}{MnLargeSymbols}{'164}
\DeclareMathDelimiter{\rrangle}{\mathclose}%
                     {MnLargeSymbols}{'171}{MnLargeSymbols}{'171}
\makeatother

\definecolor{dkcyan}{rgb}{0.1, 0.3, 0.3}
\definecolor{dkgreen}{rgb}{0,0.3,0}
\definecolor{olive}{rgb}{0.5, 0.5, 0.0}
\definecolor{dkblue}{rgb}{0,0.1,0.5}

\definecolor{col:ln}{rgb}  {0.1, 0.1, 0.7}
\definecolor{col:str}{rgb} {0.8, 0.0, 0.0}
\definecolor{col:db}{rgb}  {0.9, 0.5, 0.0}
\definecolor{col:ours}{rgb}{0.0, 0.7, 0.0}

\definecolor{lightgreen}{RGB}{170, 255, 220}
\definecolor{darkbrown}{RGB}{121,37,0}

\colorlet{listing-comment}{gray}
\colorlet{operator-color}{darkbrown}
\colorlet{comment-color}{black!50}

\lstdefinelanguage{custom-lang}{
	keywords={let, in, where, match, with, when, if, then, else, for, repeat, return, to, do, from},
	keywordstyle=[1]\color{dkblue},
	morekeywords=[2]{verify, systemToNBA, LTLtoNBA, eProduct, uProduct},
	keywordstyle=[2]\color{dkgreen},
    morekeywords=[3]{underApprox, overApprox},
	keywordstyle=[3]\color{darkbrown},
	comment=[l][\color{comment-color}]{//},
	literate=%
	{=}{{{\color{operator-color}=}}}1
	{|}{{{\color{dkblue}|}}}1
	{:}{{{\color{dkblue}:}}}1
	{:=}{{{\color{dkblue}:=}}}1
    {@}{ }1
}

\lstdefinestyle{default}{
	escapeinside={(*}{*)},
	basicstyle=\ttfamily\fontsize{9.3}{10.3}\selectfont,
	columns=fullflexible,
	commentstyle=\sffamily\color{black!50!white},
	framexleftmargin=1em,
	framexrightmargin=1ex,
	keepspaces=true,
	keywordstyle=\color{dkblue},
	mathescape,
	numbers=left,
	numberblanklines=false,
	numbersep=1.25em,
	numberstyle=\relscale{0.8}\color{gray}\ttfamily,
	showstringspaces=true,
	stepnumber=1,
	xleftmargin=2em,
}

\lstnewenvironment{code}[1][]
{\small
	\lstset{
		style=default, 
		language=custom-lang,
		#1
	}
}
{}

\lstdefinelanguage{example-lang}{
	keywords={while,do},
	keywordstyle=[1]\bfseries,
	comment=[l][\color{comment-color}]{//},
	literate=%
	{<-}{{{\color{dkblue}$\leftarrow$}}}1
	{@}{ }1
}

\lstdefinestyle{example-style}{
	escapeinside={(*}{*)},
	basicstyle=\ttfamily\fontsize{8.4}{9.7}\selectfont,
	columns=fullflexible,
	commentstyle=\sffamily\color{black!50!white},
	framexleftmargin=0em,
	framexrightmargin=0ex,
	keepspaces=true,
	keywordstyle=\color{dkblue},
	mathescape,
	numbers=left,
	numberblanklines=false,
	showstringspaces=true,
	stepnumber=1,
	xleftmargin=0em,
	numbers=none
}

\lstnewenvironment{exampleCode}[1][]
{\small
	\lstset{
		style=example-style, 
		language=example-lang,
		#1
	}
}
{}

\usepackage{scalefnt}

\newcounter{claim} 
\renewcommand{\theclaim}{\arabic{claim}}

\crefname{claim}{claim}{claims}

\renewenvironment{proof}[1][Proof]{%
    \par\noindent\textit{#1.}\hspace{0.5em}\ignorespaces%
}{%
    \hfill$\qed$%
    \par\vspace{0.5em}%
}

\newcommand{\abs}[1]{\left\lvert#1\right\rvert}

\newcommand{\lang}{\mathcal{L}}
\newcommand{\dom}{\mathit{dom}}

\newcommand{\ap}{\mathtt{AP}}
\newcommand{\trace}{\gamma}

\newcommand{\traces}{\mathit{Traces}}

\newcommand{\globally}{\LTLglobally}
\newcommand{\eventually}{\LTLfinally}
\newcommand{\until}{\ U\ }

\newcommand{\nextt}{\LTLnext}

\newcommand{\architecture}{\mathtt{Arc}}
\newcommand{\proc}{\text{Proc}}
\newcommand{\plant}{\texttt{p}}
\newcommand{\plantall}{\mathbb{P}}
\newcommand{\contr}{\texttt{c}}
\newcommand{\contrall}{\mathbb{C}}
\newcommand{\env}{\texttt{e}}
\newcommand{\inputs}[1]{{I_{#1}}}
\newcommand{\outputs}[1]{{O_{#1}}}

\newcommand{\run}{\rho}
\newcommand{\runt}{\mathit{run}}
\newcommand{\runs}{\mathit{Runs}}

\newcommand{\aut}{\mathcal{A}}
\newcommand{\acc}{\Omega}
\newcommand{\safety}{\mathit{safe}}

\newcommand{\reachable}{\texttt{reachable}}

\newcommand{\strat}{\mathcal{M}}
\newcommand{\moore}{\strat}

\newcommand{\stratoutput}{\mathit{out}}

\newcommand{\labelts}{o}

\newcommand{\univcontr}{\mathcal{U}}
\newcommand{\withprophecies}{\univcontr}
\newcommand{\underkappa}{\underline{\kappa}}
\newcommand{\overkappa}{\overline{\kappa}}
\newcommand{\approximation}{\mathcal{W}}
\newcommand{\underapprox}{\underline{\approximation}}
\newcommand{\overapprox}{\overline{\approximation}}

\newcommand{\ltlToAut}{\textsf{LTLtoDSA}}

\newcommand{\synthesize}{\textsf{synthesize}}

\newcommand{\solve}{\textsf{solve}}

\newcommand{\refine}{\mathtt{refine}}
\newcommand{\learnApprox}{\mathtt{learnApprox}}
\newcommand{\learnCTL}{\mathtt{learnCTL}}

\newcommand{\compose}{\mathtt{compose}}

\newcommand{\prophecy}{\theta}

\newcommand{\prophecyall}{\mathbb{F}}

\newcommand{\bracket}[1]{\langle #1\rangle}

\newcommand{\assign}{\mathit{asgn}}
\newcommand{\overload}{\mathit{overload}}
\newcommand{\task}{\mathit{task}}
\newcommand{\busy}{\mathit{busy}}

\newcommand{\prototype}{\textsc{UCLearn}}{}
\newcommand{\unicon}{\textsc{unicon}}{}

\lstdefinelanguage{custom-lang}{
	keywords={Input, Output, let, in, match, with, when, if, then, else, elif, for, to, do, rec, return, new, not, and, while, each, break},
	keywordstyle=[1]\color{dkblue}\bfseries,
	morekeywords=[2]{append, Set, Dict, Queue, pop, push, add, contains},
	keywordstyle=[2]\sffamily,
	morekeywords=[3]{generalController, existentialProjection, composition, isMember, stepComposition, modifiedPlant, h, isRealizable, synthesize},
	keywordstyle=[3]\color{dkcyan}\ttfamily,
	comment=[l][\color{comment-color}]{//},
	literate=%
	{=}{{{\color{operator-color}=}}}1
	{<-}{{{\color{operator-color}$\leftarrow$}}}1
	{|}{{{\color{dkblue}$\mid$}}}1
	{:}{{{\color{dkblue}:}}}1
	{:=}{{{\color{dkblue}:=}}}1
	{@}{ }1
}

\lstdefinestyle{default}{
	escapeinside={(*}{*)},
	basicstyle=,
	columns=fullflexible,
	commentstyle=\sffamily\color{black!50!white},
	framexleftmargin=1em,
	framexrightmargin=1ex,
	keepspaces=true,
	keywordstyle=\color{dkblue},
	mathescape,
	numbers=left,
	numberblanklines=false,
	numbersep=0.5em,
	numberstyle=\relscale{0.75}\color{gray}\ttfamily,
	showstringspaces=true,
	stepnumber=1,
	xleftmargin=1.2em,
}

\lstnewenvironment{mycode}[1][]
{\small
	\lstset{
		style=default, 
		language=custom-lang,
		#1
	}
}
{}

\graphicspath{{figures/}} 

\begin{document}

\maketitle              
\begin{abstract}
\emph{Universal Safety Controllers (USCs)} are a promising logical control framework that guarantees the satisfaction of a given temporal safety specification when applied to any realizable plant model. Unlike traditional methods, which synthesize one logical controller over a given detailed plant model, USC synthesis constructs a \emph{generic controller} whose outputs are conditioned by plant behavior, called \emph{prophecies}. Thereby, USCs offer strong generalization and scalability benefits over classical logical controllers. However, the exact computation and verification of prophecies remain computationally challenging. 
In this paper, we introduce an approximation algorithm for USC synthesis that addresses these limitations via learning. Instead of computing exact prophecies, which reason about sets of trees via automata, we only compute under- and over-approximations from (small) example plants and infer computation tree logic (CTL) formulas as representations of prophecies. The resulting USC generalizes to unseen plants via a verification step and offers improved efficiency and explainability through small and concise CTL prophecies, which remain human-readable and interpretable. Experimental results demonstrate that our learned prophecies remain generalizable, yet are significantly more compact and interpretable than their exact tree automata representations.

\end{abstract}

\section{Introduction}\label{sec:introduction}

The automatic construction of correct-by-design systems is central to both \emph{reactive synthesis}, which derives correct implementations from temporal logic specifications, and \emph{supervisory control}, which derives correct restrictions of existing open systems (plants). In practice, control tasks often combine both perspectives: behavioral goals are specified in logic, while the physical system, i.e., robots, sensors, or environments, is modeled as a plant. This has led to a large variety of synthesis approaches that combine specification automata with plant models and solve the resulting $\omega$-regular game to derive a correct-by-design controller (see books and surveys by \citet{tabuada2009verification,belta2017formal,Review24_FMandControl_YinGaoYu}).
These approaches are, however, known to face major scalability challenges: exploring all plant behaviors in a model can lead to an intractable explosion in the state space during synthesis. 
Even more importantly, when the plant is only known at runtime, or evolves over time, then the usual state exploration becomes impractical very quickly. 

These severe limitations have motivated the introduction of \emph{Universal Safety Controllers (USCs)} by \citet{tacaspaper}, which shift the focus from plant-specific synthesis to the computation of a \emph{universal controller} derived from the specification alone, whose decisions are conditioned by so-called \emph{prophecies}. 
As a result, USCs promise two major advantages: (1) generalization: USCs provide a correct solution for all realizable plant models, and (2) computational efficiency: by focusing on the specification, USC synthesis promises to reduce the computational burden of utilizing a complete plant model.
The latter computational benefit, however, is only achieved if prophecies are small and concise. %
\citet{tacaspaper} use tree automata as prophecies which encode all possible future branching structure of the plant potentially relevant for the specification.
While preserving correctness and completeness, this technique renders both prophecy synthesis and prophecy verification computationally very intense.

This paper addresses this limitation of the automata-based approach by introducing the first \emph{learning-based} algorithm for prophecy construction.
Concretely, we obtain prophecies in computation tree logic (CTL)~\cite{DBLP:conf/lop/ClarkeE81}, which we learn from a small set of representative (nominal) plant models drawn from the application domain. The resulting USC generalizes to unseen plants via a verification step and offers improved efficiency and explainability through small and concise CTL prophecies, which remain generalizable and human-readable. %

\subsection{Motivating Example}\label{sec:introduction:example}

As a motivating example, we consider the problem of synthesizing a simple load-balancing scheduling controller, which assigns incoming tasks to two processing units (CPUs).
The load-balancing controller is specified with the following temporal formula:
\begin{equation}\label{equ:spec:exp}
 \varphi \coloneqq \LTLglobally (\task \rightarrow \LTLnext (\assign_1 \vee \assign_2))\, \wedge\, \LTLglobally \neg \overload,
\end{equation}
where $\mathit{task}$ models the arrival of a new task and is controlled by the \emph{environment}, $\assign_i$ assigns an incoming task to either $\mathit{cpu}_1$ or $\mathit{cpu}_2$ and is controlled by the controller, and $\mathit{overload}$ means that a task was assigned to a busy processor, which is controlled by the \emph{plant}.
The specification states that there should never be an $\mathit{overload}$, i.e., the controller should never assign a task to a processor that is already busy with another task.%

Ultimately, we are interested in applying the scheduler to a concrete system implementation, such as the system described by the plant model in \Cref{fig:plant}. In this particular system implementation, each of the processors is busy for exactly one time-step once assigned a task. An example of a controller that satisfies our specification for this plant, and would typically be found by a standard synthesis algorithm, is the round-robin controller, which alternates between assigning new tasks to the first and second CPU. The disadvantage of this controller, however, is that it \emph{only} works correctly for our \emph{specific} plant. By contrast, a universal controller is immediately applicable to an entire set of plants.

\begin{figure}[t]
        \centering
        \resizebox{0.9\linewidth}{!}{
        \tikzstyle{state}=[draw, circle, fill=none, minimum width=0.7cm, 
minimum height = 0.7cm,
align=center, thick]

\begin{tikzpicture}[->,>=stealth',shorten >= 1pt,auto]

\newcommand\x{2}
\newcommand\y{.7}
\newcommand\bend{15}

\node[state] (p0){%
    $~s_0$
};
\node[state] (p1) [above right = \y and \x of p0]{%
    $~s_1$
};

\node[state] (p2) [below right = \y and \x of p0]{%
    $~s_2$
};

\node[state] (p3)[right = 2*\x of p0]{%
    $~s_3$
};

\node [below left=0.1em and 0.1em of p0] (l0) {$\emptyset$};
\node [above =0.1em of p1] (l1) {$\{\busy_1\}$};
\node [below =0.1em of p2] (l2) {$\{\busy_2\}$};
\node [above right=-0.1em and -0.1em of p3] (l3) {$\overload\}$};
\node [above right=1.0em and -0.6em of p3] (l4) {$\{\busy_1,\busy_2,$};

\node[state, draw=none] (init)[left = 0.7 of p0]{%
};
\node[state, draw=none] (align)[right = 0.7 of p3]{%
};

\path (init) edge[thick] (p0)
(p0) edge[loop above,thick] node[above, align=center] {%
  $*$
  } (p0)
(p0) edge[bend left=\bend,thick] node[above,xshift = 5pt, yshift= 7pt] {%
  $\assign_1$
  } (p1)
(p0) edge[bend right=\bend,thick] node[below, yshift = -5pt] {%
  $\assign_2$
  } (p2)
(p1) edge[bend left=\bend,thick] node[above, align=center] {%
  $*$
  } (p0)
(p1) edge[bend left=\bend,thick] node[above, xshift = 10pt] {%
  $\assign_1$
  } (p3)
(p1) edge[bend left=\bend,thick] node[right] {%
  $\assign_2$
  } (p2)
(p2) edge[bend right=\bend,thick] node[below] {%
  $*$
  } (p0)
(p2) edge[bend right=\bend,thick] node[below, xshift = 10pt] {%
  $\assign_2$
  } (p3)
(p2) edge[bend left=\bend,thick] node[left] {%
  $\assign_1$
  } (p1)
(p3) edge[loop right,thick] node[right] {%
  $*$
  } ();
\end{tikzpicture}
        }
    \caption{A plant that models the occupancy of processors in the load balancer example. We use $*$ to represent edges that are taken whenever no other edge guard is satisfied.}
    \label{fig:plant}
\end{figure}
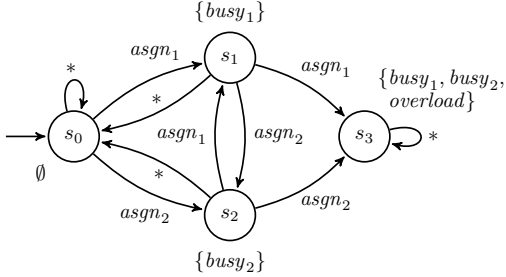

\smallskip
\noindent\textbf{Previous work: automata-based USC synthesis.}
The previous approach to universal controller synthesis, \unicon{} by \citet{tacaspaper}, takes as input only a specification: the USC is computed independently of any specific plant.  For each possible output of the controller, the synthesis algorithm computes the exact condition on the plant under which this output is correct. In our example, both $\assign_1$ and $\assign_2$ might be correct outputs, depending on whether the respective CPU is available, and also depending on the \emph{future} availability of the CPU. For example, if, in a hypothetical plant, $\mathit{cpu}_1$ would become permanently busy if used in the first step, it would only be correct to start with $\assign_2$. \unicon{} computes a tree automaton that recognizes \emph{exactly those} plants where $\assign_1$ (or $\assign_2$, respectively) does not lead to an immediate $\overload$ and where, additionally, a control strategy exists for the state reached by executing $\assign_1$ (or $\assign_2$, respectively).
The tree automaton is non-trivial and difficult to interpret; we omit depicting it here.

\smallskip
\noindent\textbf{New approach: learning-based USC synthesis.}
Our new learning-based approach, \prototype,  takes as input both the specification and a set of \emph{nominal} plants which are representative of the plants the controller will be applied to. Similar to \unicon{}, the prophecies express conditions on the plants for which a certain output should be applied. However, instead of \emph{characterizing} precisely the set of plants for which a particular output is correct, our learned prophecies \emph{separate} plants for which a control output is correct from those where the output is incorrect. This relaxation introduces a freedom of choice in the selection of the separating condition. The precise condition computed by \unicon{} is also one of the separating conditions, but typically, there are separating conditions that are much simpler, easier to verify, and easier to understand. We find such a condition, expressed as a temporal formula, using a learning algorithm for CTL.

\Cref{fig:approx:controller} shows the controller synthesized by \prototype{}~\cite{uclearn} for the specification of the load-balancing controller together with the plant from \Cref{fig:plant} as the (single) nominal plant. The controller has three states. The initial state $q_0$ represents the case that no task has been received yet. Here, $\mathit{no\_asgn}$ is a correct decision independently of the plant. The outputs $\assign_1$ and $\assign_2$ are allowed as long as the respective CPU is not busy. State $q_1$ represents the case that a task has been received and still needs to be assigned to a CPU. The difference to $q_0$ is that  $\mathit{no\_asgn}$ is no longer an option. Finally, $q_2$ represents the case that an overload has occurred. This will never happen with the plant from \Cref{fig:plant}, although it might happen on plants which the synthesis algorithm has not seen. In $q_2$, no controller output is correct, because the specification has already been violated.

The synthesized controller is not only correct for the given plant, it also generalizes to other plants that have sufficient CPU availability and use $\mathit{busy}$ like in our example.
\prototype{} thus produces controllers that are more general than the plant-specific controllers computed by classic algorithms from reactive synthesis and supervisory control, and simpler than the universal controllers produced by \unicon{}.

\begin{figure}[t]
    \newcommand{\arrow}{\twoheadleftarrow}
    \centering
    \resizebox{1\linewidth}{!}{
        \tikzstyle{state}=[draw, circle, fill=none, minimum width=0.7cm, 
minimum height = 0.7cm,
align=center, thick]

\begin{tikzpicture}[->,>=stealth',shorten >= 1pt,auto]

\newcommand\x{2}
\newcommand\y{1}
\newcommand\bend{15}

\node[state] (p0){%
    $~q_0$
};
\node[state] (p1) [right = \x of p0] {%
    $~q_1$
};

\node[state] (p2) [below right = \y and 0.5*\x of p0]{%
    $~q_2$
};

\node[state, draw=none] (init)[above left = 0.7 and 0.7 of p0]{%
};

\path (init) edge[thick] (p0)
(p0) edge[thick,bend left=\bend] node[above] {%
  $\task$
  } (p1)
(p0) edge[loop above,thick] node[above] {%
  $\neg \task$
  } (p0)
(p0) edge[thick] node[left, xshift = -0.1em] {%
  $\overload$
  } (p2)
(p1) edge[loop above,thick] node[above] {%
  $\task$
  } (p1)
(p1) edge[thick, bend left = \bend] node[below] {%
  $\neg\task$
  } (p0)
(p1) edge[thick] node[right] {%
  $\overload$
  } (p2)

(p2) edge[loop right,thick] node {%
  $*$
  } ()
;

\node [above left= 0.05 and 0.85 of p0] (l1) {${\mathit{no\_asgn}}\leftarrow \bracket{\top}$};
\node [above left= -0.4 and 0.2 of p0] (l0) {${\assign_1}\leftarrow \bracket{\neg \busy_1}$};
\node [below left = -0.4 and 0.2 of p0] (l00) {${\assign_2}\leftarrow \bracket{\neg \busy_2}$};
\node [above right= 0.05 and 0.2 of p1] (l11) {${\mathit{no\_asgn}}\leftarrow \bracket{\bot}$};
\node [above right= -0.4 and 0.2 of p1] (l0) {${\assign_1}\leftarrow \bracket{\neg \busy_1}$};
\node [below right = -0.4 and 0.2 of p1] (l00) {${\assign_2}\leftarrow \bracket{\neg \busy_2}$};

\node [left = 0em of p2, xshift=-3em,yshift=2em] (p22) {};

\node [below=0.2em of p22, xshift=-0.2em] (l22) {${\mathit{no\_asgn}}\leftarrow \bracket{\bot}$};
\node [below=1.2em of p22] (l2) {$\assign_1 \leftarrow\bracket{\bot}$};
\node [below=2.2em of p22] (l22) {$\assign_2\leftarrow\bracket{\bot}$};

\end{tikzpicture}
    \vspace*{0.5cm}}
    \caption{An approximate USC with CTL prophecies.} 
    \label{fig:approx:controller}
\end{figure}
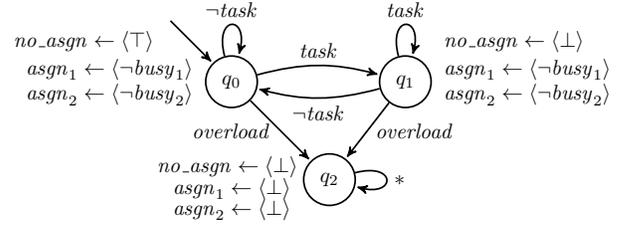

\subsection{Overview}

\Cref{fig:overview:uclearn} gives an overview of \prototype{}. The general idea is that an approximate USC is initially obtained from the specification and then continuously refined based on the plants provided to the algorithm.
An approximate USC provides for each state and controller output an over- and an under-approximation of the prophecy: the under-approximation $\underline{\kappa}$ characterizes plants for which the output is guaranteed to be correct, the over-approximation $\overline{\kappa}$ characterizes plants for which the output is not guaranteed to be incorrect. As shown in the upper part of \Cref{fig:overview:uclearn}, we initialize the process by translating the safety LTL formula to an automaton, and setting $\underline{\kappa}=\emptyset$, i.e., the controller output is guaranteed to be correct for the empty set, and  $\overline{\kappa}=\plantall$, i.e., the output is not guaranteed to be incorrect for all plants. %

The main refinement loop, framed in red in \Cref{fig:overview:uclearn}, applies the current approximation $\mathcal{W}$ to a new nominal plant and uses the result to improve $\mathcal{W}$. We first concretize the universal controller to an explicit controller for the new plant by choosing the outputs on which the prophecy conditions are satisfied. We then verify if the explicit controller is correct for the plant: if yes, $\mathcal{W}$ already covers the plant and there is no need to refine. If the controller is incorrect, we synthesize a correct controller and obtain positive and negative samples based on the states in the plant where the controller acts correctly and incorrectly, respectively. A separate CTL learning algorithm ($\learnCTL$) then finds a formula that separates the new sets of positive/negative samples.

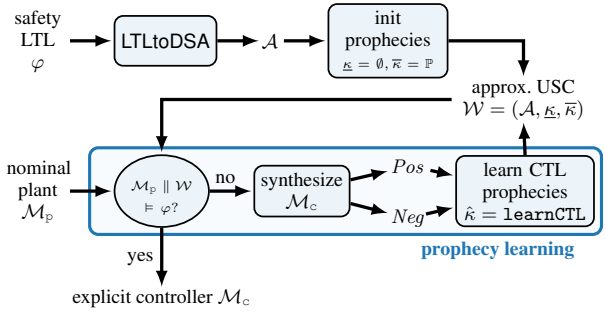
\begin{figure}[t]
    \centering
    \resizebox{0.92\textwidth}{!}{
    \begin{minipage}[c]{1.1\textwidth}
        \usetikzlibrary{trees,decorations,arrows,automata,shadows,positioning,plotmarks,backgrounds,shapes}
\usetikzlibrary{calc,matrix,fit,petri,decorations.markings,decorations.pathmorphing,patterns,intersections,decorations.text}

\tikzstyle{state}=[draw, rectangle, minimum width=.8cm, 
minimum height = .8cm, align=center, thick, draw=black!80]
\tikzstyle{emptiness}=[draw, ellipse, minimum width=.8cm, 
minimum height = .8cm, align=center, thick, draw=black!80]
\tikzstyle{nocorner}=[draw = none, rectangle, minimum width=0cm, 
minimum height = 0cm, align=center, thick, inner sep=0]
\usetikzlibrary{fit, shapes.geometric}
\newcommand\x{0.7}
\newcommand\y{1.4}
\definecolor{diagramcolor}{RGB}{31,119,180}
\colorlet{red}{diagramcolor}
\newcommand\boxcolor{diagramcolor}
\newcommand\statefill{diagramcolor!8}
\tikzset{
  state/.append style={fill=\statefill, draw=black},
  emptiness/.append style={fill=\statefill, draw=black},
  nocorner/.append style={draw=none}
}
\begin{tikzpicture}[->,>=stealth',shorten >= 1pt,scale=0.8]
\begin{footnotesize}
\node[nocorner,text width=1cm,text centered] (p0) at (0,0) {%
safety\\ LTL\\
  $\varphi$%
};%
\node[state,rounded corners, fill=\statefill] (p1) [right = \x of p0] {%
  $\ltlToAut$%
  };%

\node[nocorner] (p2) [right = \x of p1] {%
  $\mathcal{A}$
  };%

\node[state,rounded corners, fill=\statefill,text width=1.7cm,text centered] (p3) [right = \x of p2] {%
  init\\ prophecies\\
 \tiny{ $\underkappa = \emptyset, \overkappa = \plantall$}
  };%
\node[nocorner] (p4) [below right = 0*\y and 0.3*\x of p3] {%
  approx. USC\\
  $\approximation = (\aut, \underkappa, \overkappa)$
  };%
\node[nocorner,text width=1cm,text centered] (t0) [below = 1*\y of p0] {%
nominal\\ plant\\
$\strat_\plant$%
  };%
\node[state, draw = none, fill=none] (t1draw) [right = 0.6*\x of t0]
{};%
\node[emptiness,fill=\statefill,inner sep=1pt,text width=1cm] (t1) [right = \x of t0] {%
 \tiny{ $\strat_\plant \parallel \approximation$ $\vDash \varphi ?$}
  };

\node[state,rounded corners, fill=\statefill] (t2) [right = \x of t1] {%
synthesize\\
$\strat_\contr$%
  };

\node[nocorner] (t3) [above right = -0.1*\y and 0.9*\x of t2] {%
$\mathit{Pos}$
  };

  \node[nocorner] (t4) [below right = -0.1*\y and 0.9*\x of t2] {%
 $\mathit{Neg}$%
  };

\node[state,rounded corners, fill=\statefill,text width=2cm,text centered] (t5) [right = 2.4*\x of t2] {%
  learn CTL\\ prophecies\\
  $\hat{\kappa} = \learnCTL$%
  };

\node[nocorner] (u2) [below = .7*\y of t1] {%
  explicit controller 
 $\strat_\contr$%
  };

\node[nocorner] (u4) [below = .15*\y of t5, xshift=-4mm] {%
\textbf{\textcolor{\boxcolor}{prophecy learning}}
};

\begin{pgfonlayer}{background}
\node (box1) [state,rounded corners, fill=red!3, fit={(t1)(t5) (t4) (t3) (t1draw)}, draw=\boxcolor, line width = 0.05cm, yshift=0] {};
\end{pgfonlayer}

\node[coordinate, draw = none] (draw1) [below = 0.3 of p4]{};

\draw [thick,line width=1.5pt,-latex](p3.east) -|  (p4.north);
\draw [thick,line width=1.5pt,-latex](t5.north) --  (p4.south);
\draw [thick,line width=1.5pt,-latex]($(p4.west)+(-0.2,0)$) -|  (t1.north);

\path (p0) edge[-latex,line width=1.5pt] node[above] {%
       $ $
      }  (p1);

\path (p1) edge[thick, -latex,line width=1.5pt] node[above] {%
       $ $
      }  (p2);
      
\path (p2) edge[thick, -latex,line width=1.5pt] node[above] {%
       $ $
      }  (p3);

\path (t0) edge[thick, -latex,line width=1.5pt] node[above] {%
       $ $
      }  (t1);

\path (t1) edge[thick, -latex,line width=1.5pt] node[above, xshift=-1mm] {%
       no
      }  (t2);
  \path (t1) edge[thick, -latex,line width=1.5pt] node[left, xshift=-0mm] {%
       yes
      }  (u2);
\path (t2) edge[thick, -latex,line width=1.5pt] node[above] {%
       $ $
      }  (t3);
\path (t2) edge[thick, -latex,line width=1.5pt] node[above] {%
       $ $
      }  (t4);
\path (t3) edge[thick, -latex,line width=1.5pt] node[above] {%
       $ $
      }  (t5);
\path (t4) edge[thick, -latex,line width=1.5pt] node[above] {%
       $ $
      }  (t5);

\end{footnotesize}

\end{tikzpicture}
    \end{minipage}
    }
\caption{An overview of UCLearn.}
\label{fig:overview:uclearn}
\end{figure}

\subsection{Related Work}

The generalization and adaptability of the universal controller stem from its \emph{permissiveness}—its ability to represent a set of control strategies rather than a single one. Permissiveness has been widely studied in supervisory control~\cite{cassandras2021introduction} and reactive synthesis~\cite{bernet2002:permissiveStrategies,bouyer2011:measuringpermissiveness,FremontSeshi_Improvisation,Klein2015:mostGeneralController,AnandNS23}, though typically under a fixed plant model with limited adaptability. Other works consider strategies correct for sets of plants—e.g., dominant~\cite{DammF14,FinkbeinerP22} and admissible~\cite{Berwanger07,BassetRS17} strategies, or restricted plant classes~\cite{AnandMNS23,BrenguierRS17}—but they yield a single strategy correct for all plants. In contrast, our method synthesizes a controller that adapts its strategy to the given plant model.

Temporal logic specification learning has been well studied~\cite{neider2018learning,raha2023synthesizing,valizadeh2024ltl,pommellet2024sat}, but these works focus on inferring specifications from samples, not using them for control. Learning has been applied to LTL synthesis without plant models~\cite{KretinskyMPZ25,BalachanderFR23}, whereas we learn CTL formulas that characterize plant behavior.
We also leverage prophecy variables~\cite{AbadiL91,BeutnerF22} to anticipate future plant actions during universal controller synthesis, following the framework introduced in~\cite{tacaspaper}.

\section{Preliminaries}
\noindent\textbf{System Architectures.}
For every set $X$, we write $X^*$ and $X^\omega$ to denote the sets of finite and respectively infinite, sequences of elements of $X$, and let $X^\infty = X^* \cup  X^\omega$. 
For $\run\in X^\infty$, we denote $\abs{\run}\in\nat\cup \{\infty\}$ the length of $\run$, and define $\dom(\run)\coloneqq \{0,1,\ldots,\abs{\run}-1\}$.
For $\run = x_0x_1\cdots\in X^\infty$ and $i\in\dom(\run)$, we write $\run[i]$ to denote its $i$-th element, i.e., $x_i$.
For $\run\in X^\infty$ and some set $Y$, we write $\run\downarrow_Y$ to denote the sequence $\run'$ with $\run'[i] = \run[i]\cap Y$.
We focus on systems with three processes, $\proc = \{\contr,\plant,\env\}$, i.e., a controller, a plant, and an environment and $\ap = \outputs{\env}\uplus\outputs{\contr}\uplus\outputs{\plant}$ partitioned into the output propositions of the processes.
This yields an architecture $\architecture = (\inputs{\env},\outputs{\env},\inputs{\contr},\outputs{\contr},\inputs{\plant},\outputs{\plant})$ over $\ap$ where the input propositions of a process $i\in\proc$ are output propositions of all other processes, i.e., $\inputs{i} = \bigcup_{j\neq i}\outputs{j}$.

\smallskip
\noindent\textbf{Traces and Specifications.}
We fix alphabet $\Sigma = 2^{\ap}$ to the power set of the atomic propositions $\ap$.
A \emph{trace} over $\Sigma$ is a sequence $\trace\in \Sigma^\omega$.
A \emph{specification} $\varphi$ specifies some restrictions on the traces. 
We write $\trace \vDash\varphi$ to denote that the trace $\trace$ satisfies the specification $\varphi$.
The language $\lang(\varphi)$ of a specification $\varphi$ represents the set of traces that satisfy~$\varphi$.
Linear-time temporal logic~(LTL)~\cite{Pnueli77} specifications over atomic propositions $\ap$ are defined by 
$ \varphi, \psi ::= \alpha\in\ap ~ | ~ \neg \varphi ~ | ~ \varphi \lor \psi ~ | ~ \varphi \land \psi ~ | ~ \nextt \varphi ~ | ~ \varphi \until \psi ~ | ~ \eventually \varphi ~ | ~ \globally \varphi$.
Computation tree logic (CTL)~\cite{DBLP:conf/lop/ClarkeE81} specifies temporal logics over trees by differentiating between path formulas and state formulas. Paths are quantified with $\forall \varphi$ ($\varphi$ holds on all paths) and $\exists \varphi$ ($\varphi$ holds on some path).
Additionally, every temporal modality must be preceded by a path quantifier.
A system is said to satisfy an LTL or CTL specification $\varphi$ as defined by the usual semantics; we refer to~\cite{baier2008principles} for more details. 

\smallskip
\noindent\textbf{Process Strategies.}
We model a strategy for process ${i}$ as a Moore machine~$\moore = (S, s_0,\tau,\labelts)$ over inputs $\inputs{i}$ and outputs $\outputs{i}$, consisting of a finite set of states $S$, an initial state $s_0 \in S$, a transition function $\tau: S \times 2^\inputs{i} \rightarrow S$ over inputs, and an output labeling function $\labelts: S \rightarrow 2^\outputs{i}$.
For a finite input sequence $\trace = \alpha_0\alpha_1\cdots\alpha_{k-1}\in (2^{\inputs{i}})^*\!$, $\moore$ produces a finite path $s_0s_1\cdots s_k$ and an output sequence $\labelts(s_0)\labelts(s_1)\cdots\labelts(s_k)\in (2^{\outputs{i}})^*$ such that $\tau(s_j, \alpha_j) = s_{j+1}$.
We write $\stratoutput(\strat,\trace)$ to denote $\labelts(s_k)$.
Similarly, for an infinite input sequence $\trace\in (2^{\inputs{i}})^\omega\!$, $\strat$ produces an infinite path $s_0s_1\cdots$ and an infinite output sequence $\labelts(s_1)\labelts(s_1)\cdots\in (2^{\outputs{i}})^\omega$.
$\traces(\strat)$ denotes the set of all traces $\trace\in\Sigma^\omega$ such that for input sequence $\trace\downarrow_{\inputs{i}}$, $\strat$ produces the output sequence $\trace\downarrow_{\outputs{i}}$.
We say that $\strat$ satisfies a specification~$\varphi$, denoted $\strat \models \varphi$, if $\trace \models \varphi$ holds for all trace $\trace\in\traces(\strat)$.
For a strategy $\moore = (S, s_0,\tau,\labelts)$, we write $\moore(s)$ to denote the strategy $(S, s,\tau,\labelts)$ with the same transition and output labeling function but with initial state $s$.
The set of all plant strategies and controller strategies are $\plantall$ and $\contrall$, respectively.
The parallel composition $\moore_{i} \parallel \moore_j$ of two strategies $\moore_{i} = (S_i,s^i_0,\tau_i,o_i)$, $\moore_j = (S_j,s^j_0,\tau_j,o_j)$ of processes $i, j\in\proc$ is a strategy, i.e., a Moore machine, $(S,s_0,\tau,o)$ over inputs $(\inputs{i} \cup \inputs{j})\setminus(\outputs{i}\cup \outputs{j})$ and outputs $\outputs{i} \cup \outputs{j}$ with $S = S_i \times S_j$, $s_0=(s^i_0,s^j_0)$, 
$\tau((s,s'),\sigma) = (\tau_i(s,(\sigma \cup \labelts_j(s')) \cap \inputs{i}),\tau_j(s',(\sigma \cup \labelts_i(s))\cap\inputs{j}))$, and 
$\labelts((s,s')) = \labelts_i(s) \cup \labelts_j(s')$.

\smallskip
\noindent\textbf{$\omega$-Automata.}
An $\omega$-\emph{automaton} $\aut$ over alphabet $\Sigma$ is a tuple $(Q,q_0,\delta,\Omega)$ consisting of a finite set of states $Q$, an initial state $q_0\in Q$,
a transition function $\delta \colon Q\times \Sigma \rightarrow Q$, and an acceptance condition $\acc\subseteq Q^\omega$.
The unique \emph{run} of $\aut$ from state $q$ on some trace $\trace\in \Sigma^\infty$, denoted by $\runt(\aut,q,\trace)$, is a sequence of states $\run\in Q^\infty$ with $\abs{\run} =\abs{\trace}+1$, $\run[0] = q$, and $\delta(\run[i],\trace[i]) = \run[i+1]$ for all $i\in\dom(\trace)$. 
A run $\run$ is \emph{accepting} if $\run\in \acc$.
The language $\mathcal{L}(\aut)$ is the set of all traces $\trace$ for which the unique run $\runt(\aut,q_0,\trace)$ is accepting.
Furthermore, we write $\aut(q)= (Q,q,\delta,\Omega)$ to denote the automaton with the same transition function and acceptance condition but with initial state $q$, and $\reachable(\aut)$ to denote the set of states that are reachable from the initial state $q_0$ in $\aut$, i.e., there exists a run $\run$ of $\aut$ with $\run[0] = q_0$ and $\run[k] = q$ for some $k\geq 0$.

In safety automata, the acceptance $\acc = \safety(F)$ is given by a set $F\subseteq Q$ of safe states, containing the set of runs which only visits states in $F$.
An LTL specification $\varphi$ is said to be \emph{safety} if it can be expressed as a safety automaton $\aut$.
A standard result in the literature is that every safety LTL specification can be translated into an equivalent deterministic safety automaton with a double exponential blow-up~\cite{safeLtl}.
Let us denote this procedure by $\ltlToAut(\varphi)$.

Given an automaton $\aut = (Q,q_0,\delta,\acc)$ and a Moore machine $\moore$ over inputs $I$ and outputs $O$ with $2^{I\cup O}\subseteq \Sigma$, 
we define the composition $\aut \times \moore = (Q\times S, (q_0,s_0), \delta',\acc')$ as a product automaton with acceptance condition $\acc'$ given by the set of runs whose first component is accepting by $\acc$, and a partial transition function 
$\delta': (Q\times S)\times \Sigma \rightarrow (Q\times S)$ such that $\delta'((q,s),\sigma)$ is defined if and only if $\sigma\cap O = \labelts(s)$ and in that case $\delta'((q,s),\sigma) = (\delta(q,\sigma),\tau(s,\sigma))$.
We write $\runs(\aut,q,\moore)$ to denote the set of runs $\run$ of $\aut$ for which there exists a corresponding run $\run'$ of $\aut\times\moore$ with $\run'[i] = (\run[i],s_i)$ for each $i\geq 0$.
We write $\reachable(\aut\times\moore)$ to denote the set of states $(q,s)$ that are reachable from the initial state $(q_0,s_0)$ in $\aut\times\moore$, i.e., there exists a run $\run$ of $\aut\times\moore$ with $\run[0] = (q_0,s_0)$ and $\run[k] = (q,s)$ for some $k\geq 0$.

\smallskip
\noindent\textbf{Logical Controller Synthesis.}
Logical controller synthesis addresses the problem of, given a specification and a plant, synthesizing a controller that satisfies the specification if executed with the plant. 
For an LTL specification $\varphi$ over $\ap$, and a fixed plant strategy $\strat_\plant$, the goal is to find a controller strategy $\strat_\contr$ such that $\strat_\contr \parallel \strat_\plant \vDash \varphi$.
A plant strategy is \emph{admissible} if such a controller exists.

\section{Universal Controller Synthesis}\label{sec:universal-controller}
In this section, we recall the notion of universal controllers and their synthesis from \cite{tacaspaper}. 
Traditional controller synthesis requires exploring the state space of a given plant. In contrast, we consider a synthesis framework that abstracts away from explicit plants by reasoning over sets of possible plants, called \emph{prophecies}.
A prophecy captures assumptions about the future plant behavior. Instead of solving the synthesis problem for each plant individually, we synthesize a controller that conditions its behavior on verified prophecies for the given plant. This leads to the notion of a \emph{prophecy-annotated controller}\footnote{Note that~\cite{tacaspaper} refers to prophecy-annotated controllers as \emph{universal controllers}. We call a prophecy-annotated controller \emph{universal} iff it is correct and most permissive.} (\emph{prophecy controller} for short), which conditions controller behavior with prophecies.

\begin{definition}[Prophecy-annotated controller] 
    A \emph{prophecy} $\prophecy\subseteq\plantall$ is a set of plants. 
    A \emph{prophecy-annotated controller} $\withprophecies$ over an alphabet $\Sigma$ and a set of prophecies $\prophecyall$ is a tuple $(S, s_0, \tau, \kappa)$, where $S$ is a finite set of states, $s_0 \in S$ is the initial state, $\tau: S \times \Sigma \rightarrow S$ is the transition function, $\kappa: S \times 2^{\outputs{\contr}} \rightarrow \prophecyall$ is the prophecy annotation. 
\end{definition}

Note that, similar to $\omega$-automata, the transition function $\tau$ is defined over all possible alphabets, and hence, $\withprophecies\times\strat$ is well-defined for any strategy $\strat$.
Furthermore, each controller output from a state is associated with a prophecy that must hold for that output to be valid. For an explicit plant, we compute a consistent controller by verifying the prophecies at each state.

\begin{definition}[Consistency] 
    Given a prophecy controller $\withprophecies = (S,\allowbreak s_0,\tau,\kappa)$ and a plant $\strat_\plant$, a controller $\strat_\contr = (S^\contr, s_0^\contr,\tau^\contr,\labelts^\contr)$ is said to be consistent with $\withprophecies$ w.r.t. $\strat_\plant$, denoted by $\strat_\contr\vDash\withprophecies\parallel \strat_\plant$, 
    if for all $(s,s^\plant,s^\contr)\in \reachable(\withprophecies\times(\strat_\contr\parallel\strat_\plant))$, it holds that $\strat_\plant(s^\plant)\in \kappa(s,\labelts^\contr(s^\contr))$.
\end{definition}

We are interested in prophecy controllers that are correct when combined with a consistent controller for a plant.

\begin{definition}[Correctness] 
    Given a specification~$\varphi$, a prophecy controller $\withprophecies$ is \emph{correct} for a plant strategy $\strat_\plant$, if it holds that $\strat_\contr\vDash \withprophecies\parallel \strat_\plant$ implies $\strat_\plant \parallel \strat_\contr \vDash \varphi$.
\end{definition}

It is possible that there does not exist a controller $\strat_\contr$ that is consistent with the prophecy controller $\withprophecies$ for a given plant $\strat_\plant$. In that case, we say $\withprophecies$ is \emph{incompatible} with $\strat_\plant$, denoted by $\withprophecies\parallel\strat_\plant = \emptyset$.
If all correct controllers for an (admissible) plant are consistent with the universal controller, we call such a universal controller \emph{most permissive}.

\begin{definition}[Permissiveness]\label{def:permissive:of:universal}
    Given a specification~$\varphi$, a prophecy controller $\withprophecies$ is \emph{most permissive} for a plant strategy $\strat_\plant$, if it holds that $\strat_\plant \parallel \strat_\contr \vDash \varphi$ implies $\strat_\contr\vDash \withprophecies\parallel \strat_\plant$.
\end{definition}

The goal of universal controller synthesis is to construct a prophecy controller that is correct and most permissive for every plant.
This implies it would produce correct controllers for all admissible plants and would be incompatible for all inadmissible plants. 
However, note that a trivial solution to this problem is a prophecy controller that includes all plants in the prophecies from safe states and excludes all plants in the prophecies from unsafe states.
This is due to the fact that every plant (including inadmissible ones) would satisfy the prophecies as long as the composition is in the safe region, and would become incompatible once the composition reaches an unsafe state.
To avoid such trivial solutions, we require that if a plant satisfies a prophecy from a state, it must be compatible with the controller from that state.
This can be formalized as follows: a prophecy controller $\withprophecies = (S,s_0,\tau,\kappa)$ is said to be \emph{forward-complete} if for every plant $\strat_{\plant}$, if $\strat_{\plant}\in\kappa(s,\alpha)$ for some $s\in S$ and $\alpha\in 2^{\outputs{\contr}}$, then $\withprophecies(s)\parallel \strat_{\plant} \neq \emptyset$.
With this definition, we can now define a universal controller.
\begin{definition}[Universal Controller]
    A \emph{universal controller} $\univcontr$ for a specification $\varphi$ is a forward-complete prophecy controller that is correct and most permissive for every plant.
\end{definition}

In \cite{tacaspaper}, the authors present a procedure to synthesize a universal safety controller (USC), i.e., a universal controller for safety LTL specifications.
The procedure is based on the safety automaton of the specification and shows that the correct safe prophecies can be formulated as tree automata~\cite{automata-book}.
Furthermore, it guarantees that the USC uses the same transition structure as its safety automaton~{\cite{tacaspaper}}.
Hence, from now on, we only consider and refer to USCs (or any prophecy controller) as its automaton with prophecy annotations, i.e., $\withprophecies = (\aut, \kappa)$.

\section{Approximations of Universal Controllers}\label{sec:approximations}
The construction of a USC in~\cite{tacaspaper} computes the exact set of correct prophecies, ensuring correctness and most permissiveness for \emph{all} plants.
In practice, however, we typically work with a restricted class of plants, and building the (often complex) USC is not practical.
In this paper, we consider approximations of the USC that maintain both under- and over-approximations of the prophecies.
These approximations can be used to generate positive and negative examples for learning the prophecies.

\subsection{Approximations}
Our goal is to construct prophecy controllers that provide both sound under- and over-approximations of the USC's prophecies. 
We formalize this notion below.
\begin{definition}[Approximations]\label{def:approximation}
Given a USC $\univcontr = (\aut,\kappa)$ for a safety LTL specification, an \emph{approximation} is a tuple $\approximation = (\aut,\underkappa,\overkappa)$, where $\underkappa$ and $\overkappa$ are prophecy annotations such that: $\underkappa(q,\alpha) \subseteq \kappa(q,\alpha)\subseteq \overkappa(q,\alpha)$ for all states $q$ and propositions $\alpha\subseteq{\outputs{\contr}}$.
\end{definition}
Here, $\underkappa$ is an under-approximation and $\overkappa$ is an over-approximation of the prophecy annotations $\kappa$ of the USC $\univcontr$.
We refer to under-approximation and over-approximation as $\underapprox = (\aut,\underkappa)$ and $\overapprox = (\aut,\overkappa)$, respectively.
Note that the above definition does not require the approximations to be correct or most permissive for all plants.
Nevertheless, it is straightforward to see that the under-approximation remains correct for every plant, while the over-approximation is most permissive for every plant.

\begin{lemma}\label{rem:approximation}
    For an approximation $\approximation = (\aut,\underkappa,\overkappa)$, its under-approximation $\underapprox = (\aut,\underkappa)$ is correct and its over-approximation $\overapprox = (\aut,\overkappa)$ is most permissive for all plants.
\end{lemma}

\subsection{Refinements}
While \cref{rem:approximation} guarantees that any under-approximation is correct and any over-approximation is most permissive, such approximations are not immediately useful in practice.
For example, a trivial approximation sets $\underkappa$ to the empty set and $\overkappa$ to the set of all plants, which yields no useful insight about the actual plants.
In practical scenarios, approximations become valuable when the synthesis procedure targets a specific class of plants:
For every new plant for which we synthesize a solution, we can refine the approximation $\approximation$.
This incrementally increases the precision, targeting an approximation that is correct and permissive for this exact class of plants.
In this scenario, we assume that the general structure of the synthesis problem remains consistent, while specific aspects, such as the arrangement of obstacles, vary between plants.
For example, a refinement of the approximate controller in our running example (see \Cref{fig:approx:controller}) could be that some processor becomes unavailable under some circumstances.
Once this plant is observed, the approximations are refined to check exactly for this behavior. 
While the controller strategies remain largely consistent across different plants, the associated prophecies may vary.
To address this, we apply \emph{refinements} to update an existing approximation so that it remains correct and most permissive for the current plant.
To guarantee that each refinement increases the precision of the approximation, we formalize the following necessary conditions.

\begin{definition}[Refinement]
    Given an approximation $\approximation = (\aut,\underkappa,\overkappa)$ and a plant $\strat_\plant$, a \emph{refinement} is a different approximation $\approximation' = (\aut,\underkappa',\overkappa')$ such that: 
    \begin{enumerate}[label=(\roman*)]
        \item $\underkappa(q,\alpha) \subseteq \underkappa'(q,\alpha)$ and $\overkappa'(q,\alpha) \subseteq \overkappa(q,\alpha)$ for all states $q$ and propositions $\alpha\subseteq{\outputs{\contr}}$,\label{item:refinement-monotonicity}
        \item $\underapprox'$ and $\overapprox'$ are correct and most permissive for $\strat_\plant$.\label{item:refinement-correctness}
    \end{enumerate}
\end{definition}
The definition of refinement implies that the approximation $\approximation$ is correct or permissive for at least one more plant.
We claim that, in practice, every refinement steps captures a multitude of new plants.

\subsection{Computing Refinements}
We continue by providing algorithms for the computation of refinements.
The goal of the refinement is to incorporate the information of the plant $\strat_\plant$ into the prophecy approximations.
We do so by extracting all the sub-plants of the given plant and updating each prophecy annotation by adding or removing them based on their correctness.

\begin{algorithm}[t]
    \caption{Refinement of an Approximation}
    \label{alg:refinement}
    \begin{mycode}  
Input: An approximation $\approximation = (\aut,\underkappa,\overkappa)$ with $\aut = (Q,q_0,\Sigma,\delta,\acc)$ and a plant $\strat_{\plant} = (S^\plant,s_0^\plant,\tau^\plant,\labelts^\plant)$.
Output: A refined approximation $\approximation' = (\aut,\underkappa',\overkappa')$
let $\refine (\approximation, \strat_{\plant}) :=$
    let $G  = (Q\times S^\plant, (q_0,s_0), \delta',\acc') \gets \aut\times\strat_{\plant}$ // Compose
    let $\mathit{Win} \gets \solve(G)$ // Winning states in the game
    for each $q\in Q$, $s^\plant\in S^\plant$, and $\alpha\subseteq {\outputs{\contr}}$ do // All outputs
            if $\delta'((q,s^\plant),\alpha\cup\beta)\in \mathit{Win}$ for every $\beta\in 2^{\outputs{\env}}$
                then $\underkappa(q,\alpha) \gets \underkappa(q,\alpha) \cup \{\strat_{\plant}(s^\plant)\}$ // Add plant
                else $\overkappa(q,\alpha) \gets \overkappa(q,\alpha) \setminus \{\strat_{\plant}(s^\plant)\}$// Remove plant
    return $(\aut,\underkappa,\overkappa)$
\end{mycode}
\end{algorithm}

The overall algorithm for refining an approximation is presented in \cref{alg:refinement}.
The algorithm takes an approximation $\approximation = (\aut,\underkappa,\overkappa)$ and a plant $\strat_\plant$ as input and returns a refined approximation $\approximation' = (\aut,\underkappa',\overkappa')$.
The first step is the construction of the game $G$ that ranges over the composition of the specification automaton $\aut$ and the plant $\strat_\plant$.
We solve this game to obtain the winning region and the states of the winning region $\mathit{Win}$. 
Next, based on the winning states, it iterates over all states $q\in Q$ and propositions $\alpha\subseteq {\outputs{\contr}}$ to update the prophecy annotations $\underkappa'$ and $\overkappa'$.
If a sub-plant $\strat_{\plant}(s^\plant)$, i.e., the strategy of the plant at state $s^\plant$, is a correct choice for the output $\alpha$ from state $q$, then it is added to the under-approximation $\underkappa'(q,\alpha)$.
Otherwise, if it is not a correct choice, it is removed from the over-approximation $\overkappa'(q,\alpha)$.
This ensures that the refined approximation $\approximation'$ is more precise than the original approximation $\approximation$, and it is correct and most permissive for the plant $\strat_\plant$.

\begin{theorem}\label{thm:refinement}
Given an approximation $\approximation$ and a plant $\strat_\plant$, the procedure $\refine(\approximation, \strat_\plant)$ returns a refinement.
\end{theorem}
\begin{proof}
    First, let us recall the properties of the universal safety controllers (USC) and their synthesis procedure from~\cite{tacaspaper}.
    Let $\univcontr = (\aut = (Q,q_0,\delta,\acc),\kappa)$ be the USC returned by their procedure.
    By the definition of safe prophecies, we know that $\kappa(q,\alpha) = \{\strat_\plant \mid \exists\strat_\contr\in\contrall. \stratoutput(\strat_\contr,\epsilon) = \alpha \wedge \runs(\aut,q,\strat_\plant\parallel \strat_\contr)\subseteq \acc\}$ for all states $q\in Q$ and propositions $\alpha\subseteq {\outputs{\contr}}$.
    This means that $\kappa(q,\alpha)$ contains all plant for which there exists a controller $\strat_\contr$ that outputs $\alpha$ initially from state $q$ and ensures that the run of the composition of the plant and the controller is safe, i.e., it stays in the safe regions as in $\acc$.

    Now, let $\approximation = (\aut,\underkappa,\overkappa)$ and $\approximation' = (\aut,\underkappa',\overkappa')$ be the prophecy controllers before and after the refinement, respectively.
    By construction, we know that the winning states $Win$ in the game $G$ computed in line 6 of Algorithm 1 are the states for which there exists a controller that ensures the run in the composition game is safe.
    Therefore, $\kappa(q,\alpha)$ contains all plant for which there exists a controller that outputs $\alpha$ initially from state $q$ and ensures that the run of their composition from $q$ moves to a winning state.
    This means that $\kappa(q,\alpha)$ contains all plant for which $\alpha$ ensures going to a winning state from state $q$.
    Hence, lines 9 and 10 of Algorithm 1 ensure that $\underkappa'(q,\alpha)$ includes such plants whereas $\overkappa'(q,\alpha)$ excludes plants that do not ensure going to a winning state.
    Thus, by construction, $\approximation'$ is an approximation of $\univcontr$.

    As line 9 and 10 of Algorithm 1 only add plants to $\underkappa'$ and remove plants from $\overkappa'$, we have that $\underkappa(q,\alpha) \subseteq \underkappa'(q,\alpha)$ and $\overkappa'(q,\alpha) \subseteq \overkappa(q,\alpha)$ as required by property~(i) of refinements.
    Furthermore, as $\underkappa'(q,\alpha)$ contains all sub-plants of $\strat_\plant$ which ensures going to a winning state from state $q$ for the output $\alpha$, it is correct for the plant $\strat_\plant$.
    Similarly, as $\overkappa'(q,\alpha)$ contains all sub-plants of $\strat_\plant$ which does not ensure going to a winning state from state $q$ for the output $\alpha$, it is most permissive for the plant $\strat_\plant$.
    Thus, property~(ii) of refinements holds, and hence, $\approximation'$ is a refinement of $\approximation$.
\end{proof}

One can start with a trivial approximation (e.g., $\underkappa = \emptyset$ and $\overkappa = \plantall$) and iteratively refine it with the information from the plants encountered in practice.
Due to the monotonicity of the refinement, this also ensures that, at each step, both the under- and over-approximation preserve the solutions of previous computation steps.

\section{Synthesis with Learned Prophecies}
\label{sec:learning}
We now turn approximating controllers and the algorithms of the previous section into practice by learning representations of prophecies for under- and over-approximations.
The target formalism for approximations are formulas in computation tree logic (CTL).
\begin{algorithm}[t]
    \caption{Learning Approximations}
    \label{alg:learning}
    \begin{mycode}
Input: An approximation $\approximation = (\aut,\underkappa,\overkappa)$ with $\aut = (Q,q_0,\Sigma,\delta,\acc)$ and a set of plants $\mathit{Set} \subseteq \plantall$.
Output: A prophecy controller $\withprophecies$.
let $\learnApprox(\approximation, \mathit{Set}) :=$
    for each $\strat_\plant \in \mathit{Set}$ do // For each plant in the set
        $\approximation \gets \refine(\approximation, \strat_\plant)$ // Refine the approximation
    for each $q\in Q$, $\alpha\subseteq {\outputs{\contr}}$ do
        let $Pos \gets \underkappa(q,\alpha)$ // Positive samples
        let $Neg \gets \plantall \setminus \overkappa(q,\alpha)$ // Negative samples
        let $\phi \gets \learnCTL(Pos, Neg)$ // Learn a CTL formula
        $\kappa(q,\alpha) \gets \lang(\phi)$ // Set the prophecy annotation
    return $(\aut,\kappa)$ 
\end{mycode}
\end{algorithm} 

\subsection{Learning CTL Formulas for Approximations}
Prophecies of universal controllers, and their over- and under-approximation, are sets of plants.
The algorithm presented in \cite{tacaspaper}, therefore, uses tree automata for prophecies that accept plants iff the considered output is correct for this plant.
While tree automata represent prophecies precisely, and thereby preserve correctness and most permissiveness, they are computationally hard to compute, solve, and especially hard to understand.
We choose a more lightweight formalism for categorizing trees: CTL formulas.
While CTL formulas are not as expressive as tree automata~\cite{baier2008principles}, they are concise, human-readable, and easy to verify.

Our algorithmic solution to constructing CTL formulas for prophecies is learning.
The key idea is to use the set of plants in the under-approximation as positive samples and the complement of the set of plants in the over-approximation as negative samples.
From these samples, we use standard CTL learning techniques~\cite{learnCTL} to learn a CTL formula that accepts the positive samples and rejects the negative samples.
This learned formula can then be used as the prophecy annotation -- the over-approximation is the set of all plants that satisfy the formula, and the under-approximation is the set of all plants that do not satisfy the formula.
The learning process is summarized in \cref{alg:learning}.

As the algorithm returns a prophecy annotation that classifies the plants in the under-approximation as positive samples and the plants not in the over-approximation as negative samples, by \cref{thm:refinement} and property~\ref{item:refinement-correctness} of refinements, the resulting prophecy controller is guaranteed to be correct and most permissive for every plant in the set. 
We formalize this result in the following.
\begin{corollary}\label{thm:learning}
    Given an approximation $\approximation = (\aut,\underkappa,\overkappa)$ and a plant set $Set \subseteq \plantall$, the procedure $\learnApprox(\approximation, Set)$ returns a prophecy controller $\withprophecies = (\aut,\kappa)$ such that: $\underkappa(q,\alpha) \subseteq \kappa(q,\alpha) \subseteq \overkappa(q,\alpha)$ for all states $q$ and propositions $\alpha\subseteq{\outputs{\contr}}$.
    Furthermore, $\withprophecies$ is correct and most permissive for every plant in $Set$. 
\end{corollary}

\begin{algorithm}[t]
    \caption{On-the-fly Composition}
    \label{alg:compose}
    \begin{mycode}
Input: A prophecy controller $\withprophecies = (\aut,\kappa)$ with $\aut = (Q,q_0,\Sigma,\delta,\acc)$ and a plant $\strat_\plant = (S,s_0,\tau,\labelts)$.
Output: An explicit controller $\strat_\contr$.
let $\compose(\withprophecies, \strat_\plant) :=$
    let $s_0' \gets (q_0,s_0)$; $S' \gets \{s_0'\}$; $\tau' \gets []$; $\labelts' \gets []$ // Initialize
    let $queue \gets Queue(s_0')$ // Initialize the queue
    while $(q,s) \gets queue.pop()$ do
        for each $\alpha\subseteq {\outputs{\contr}}$ do // For each output
            if $\strat_\plant(s) \in \kappa(q,\alpha)$ then // If prophecy is satisfied
                let $\labelts'(q,s) \gets \alpha$ // Set the label
                for each $\beta\subseteq\inputs{\contr}$ do // For each input, add transition
                    $q'' \gets \delta(q,\alpha\cup\beta)$; $s'' \gets \tau(s,(\alpha\cup\beta)\cap\inputs{\plant})$ 
                    $\tau'((q,s),\beta) \gets (q'',s'')$ 
                    if $(q'',s'') \notin S'$ then // Add new states
                        $S' \gets S' \cup \{(q'',s'')\}$; $queue.push((q'',s''))$
                break // break the loop as we found a valid output
    return $(S', s_0', \tau', \labelts')$ 
\end{mycode}
\end{algorithm}

\subsection{Composition with Prophecy Controllers}
\label{sec:composition}
The goal of \cref{alg:learning} is to obtain a prophecy controller that correctly approximates the USC restricted to the class of plants similar to the ones in the set $\mathit{Set}$, i.e., the prophecy controller that is correct and most permissive for this class of plants.
Once the prophecy controller is learned, we can compose the learned prophecy controller with a concrete plant to obtain an explicit controller for the given plant, as long as the plant is contained in the learned prophecies.
The composition algorithm is presented in \cref{alg:compose}.

The algorithm implements an on-the-fly exploration of the state space of the composition of prophecy controller and plant, and only adds transitions that satisfy the prophecy annotations.
For every visited state and every possible output, it checks whether the prophecy annotation is satisfied by the plant.
If satisfied, the output is added as a label to the current state, and transitions to the next states are added for each possible input of the controller.
Once all states have been explored, the algorithm returns the explicit controller.

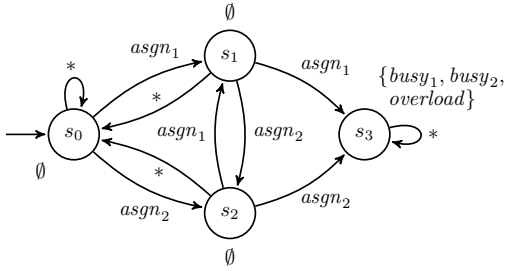
\begin{figure}[]
        \centering
        \resizebox{0.9\linewidth}{!}{
        \tikzstyle{state}=[draw, circle, fill=none, minimum width=0.7cm, 
minimum height = 0.7cm,
align=center, thick]

\begin{tikzpicture}[->,>=stealth',shorten >= 1pt,auto]

\newcommand\x{2}
\newcommand\y{.7}
\newcommand\bend{15}

\node[state] (p0){%
    $~s_0$
};
\node[state] (p1) [above right = \y and \x of p0]{%
    $~s_1$
};

\node[state] (p2) [below right = \y and \x of p0]{%
    $~s_2$
};

\node[state] (p3)[right = 2*\x of p0]{%
    $~s_3$
};

\node [below left=0.1em and 0.1em of p0] (l0) {$\emptyset$};
\node [above =0.1em of p1] (l1) {$\emptyset$};
\node [below =0.1em of p2] (l2) {$\emptyset$};
\node [above right=-0.1em and -0.1em of p3] (l3) {$\overload\}$};
\node [above right=1.0em and -0.6em of p3] (l4) {$\{\busy_1,\busy_2,$};

\node[state, draw=none] (init)[left = 0.7 of p0]{%
};
\node[state, draw=none] (align)[right = 0.7 of p3]{%
};

\path (init) edge[thick] (p0)
(p0) edge[loop above,thick] node[above, align=center] {%
  $*$
  } (p0)
(p0) edge[bend left=\bend,thick] node[above,xshift = 5pt, yshift= 7pt] {%
  $\assign_1$
  } (p1)
(p0) edge[bend right=\bend,thick] node[below, yshift = -5pt] {%
  $\assign_2$
  } (p2)
(p1) edge[bend left=\bend,thick] node[above, align=center] {%
  $*$
  } (p0)
(p1) edge[bend left=\bend,thick] node[above, xshift = 10pt] {%
  $\assign_1$
  } (p3)
(p1) edge[bend left=\bend,thick] node[right] {%
  $\assign_2$
  } (p2)
(p2) edge[bend right=\bend,thick] node[below] {%
  $*$
  } (p0)
(p2) edge[bend right=\bend,thick] node[below, xshift = 10pt] {%
  $\assign_2$
  } (p3)
(p2) edge[bend left=\bend,thick] node[left] {%
  $\assign_1$
  } (p1)
(p3) edge[loop right,thick] node[right] {%
  $*$
  } ();
\end{tikzpicture}
        }
    \caption{A plant for the load balancer example that only signals $\busy$ when processors are overloaded. We use $*$ to represent edges that are taken whenever no other edge guard is satisfied.}
    \label{fig:plant:new}
\end{figure}

Note that checking whether the prophecy is satisfied by a plant is done by model checking the corresponding CTL formula against the plant, which can be done in polynomial time~\cite{DBLP:journals/toplas/ClarkeES86}.
This allows us to efficiently check the prophecies and compose the controller with the plant on-the-fly.

\begin{algorithm}[b]
    \caption{Synthesis via Learning and Refinement}
    \label{alg:synthesis}
    \begin{mycode}
Input: An approximation $\approximation$, its learned prophecy controller $\withprophecies = (\aut,\kappa)$ with $\aut = (Q,q_0,\delta,\acc)$, and a plant $\strat_\plant$.
Output: A correct controller $\strat_\contr$ for $\strat_\plant$.
let $\synthesize(\approximation, \withprophecies, \strat_\plant) :=$
    $\strat_\contr \gets \compose(\withprophecies, \strat_\plant)$ // Compose controller
    if $\runs(\aut,q_0,\strat_\plant \parallel \strat_\contr)\subseteq \acc$ then // Check correctness
        return $\strat_\contr$ // Return correct controller
    else
        $\withprophecies \gets \learnApprox(\approximation, \{\strat_\plant\})$ // Refine and re-learn
        return $\compose(\withprophecies, \strat_\plant)$ // Return new controller
\end{mycode}
\end{algorithm}

\subsection{Synthesis via Learning and Refinement}
Although \cref{alg:compose} can be used to synthesize an explicit controller for a given plant using a learned prophecy controller, its correctness is not immediately guaranteed: If the plant is not contained in the learned prophecies, the constructed controller might not be correct.
To overcome this limitation, we add an additional model-checking step to verify whether the synthesized controller is correct for the given plant.
If it is, we return the synthesized controller, otherwise, we refine the approximation using this plant and repeat the learning process to obtain a new prophecy controller that includes this plant.
This overall synthesis procedure is summarized in \cref{alg:synthesis} and its correctness is shown in \cref{thm:refinement} and \cref{thm:learning}.
\begin{corollary}
    Given an approximation $\approximation$ and its learned prophecy controller $\withprophecies$ for a safety LTL specification $\varphi$, and a plant $\strat_\plant$, the procedure $\synthesize(\approximation, \withprophecies, \strat_\plant)$ returns a controller $\strat_\contr$ such that $\strat_\plant \parallel \strat_\contr \vDash \varphi$.
\end{corollary}

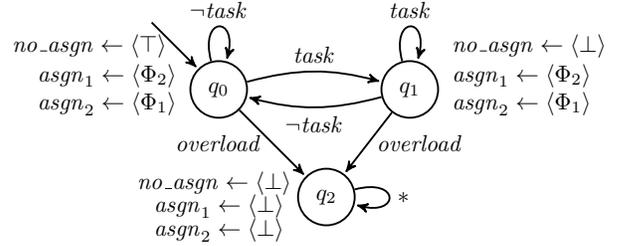
\begin{figure}[]
    \centering
    \resizebox{1\linewidth}{!}{
        \tikzstyle{state}=[draw, circle, fill=none, minimum width=0.7cm, 
minimum height = 0.7cm,
align=center, thick]

\begin{tikzpicture}[->,>=stealth',shorten >= 1pt,auto]

\newcommand\x{2}
\newcommand\y{1}
\newcommand\bend{15}

\node[state] (p0){%
    $~q_0$
};
\node[state] (p1) [right = \x of p0] {%
    $~q_1$
};

\node[state] (p2) [below right = \y and 0.5*\x of p0]{%
    $~q_2$
};

\node[state, draw=none] (init)[above left = 0.7 and 0.7 of p0]{%
};

\path (init) edge[thick] (p0)
(p0) edge[thick,bend left=\bend] node[above] {%
  $\task$
  } (p1)
(p0) edge[loop above,thick] node[above] {%
  $\neg \task$
  } (p0)
(p0) edge[thick] node[left, xshift = -0.1em] {%
  $\overload$
  } (p2)
(p1) edge[loop above,thick] node[above] {%
  $\task$
  } (p1)
(p1) edge[thick, bend left = \bend] node[below] {%
  $\neg\task$
  } (p0)
(p1) edge[thick] node[right] {%
  $\overload$
  } (p2)

(p2) edge[loop right,thick] node {%
  $*$
  } ()
;

\node [above left= 0.05 and 0.35 of p0] (l1) {${\mathit{no\_asgn}}\leftarrow \bracket{\top}$};
\node [above left= -0.4 and 0.2 of p0] (l0) {${\assign_1}\leftarrow \bracket{\Phi_2}$};
\node [below left = -0.4 and 0.2 of p0] (l00) {${\assign_2}\leftarrow \bracket{\Phi_1}$};
\node [above right= 0.05 and 0.2 of p1] (l11) {${\mathit{no\_asgn}}\leftarrow \bracket{\bot}$};
\node [above right= -0.4 and 0.2 of p1] (l0) {${\assign_1}\leftarrow \bracket{\Phi_2}$};
\node [below right = -0.4 and 0.2 of p1] (l00) {${\assign_2}\leftarrow \bracket{\Phi_1}$};

\node [left = 0em of p2, xshift=-3em,yshift=2em] (p22) {};

\node [below=0.2em of p22, xshift=-0.2em] (l22) {${\mathit{no\_asgn}}\leftarrow \bracket{\bot}$};
\node [below=1.2em of p22] (l2) {$\assign_1 \leftarrow\bracket{\bot}$};
\node [below=2.2em of p22] (l22) {$\assign_2\leftarrow\bracket{\bot}$};

\end{tikzpicture}
    \vspace*{0.5cm}}
    \caption{A refined prophecy controller for the load balancer example, where the CTL prophecy $\Phi_i$ is given by $\forall\nextt(\overload\Rightarrow\assign_i)$ for $i\in\{1,2\}$.}
    \label{fig:refined:controller}
\end{figure}

\begin{example}\label{ex:full}
    To illustrate the overall synthesis procedure, we revisit the example in \Cref{sec:introduction:example}.
    We consider the specification $\varphi$ in \cref{equ:spec:exp} and th nominal plant $\strat_\plant$ in \cref{fig:plant}.
    By using the procedure $\learnApprox(\approximation, \{\strat_\plant\})$ on this plant with an initial (coarse) approximation, we obtain the prophecy controller $\withprophecies$ in \cref{fig:approx:controller}, where prophecies are represented as CTL formulas.
    As discussed in \cref{sec:introduction:example}, this prophecy controller is not only correct for the plant $\strat_\plant$ but also generalizes to larger plants that have sufficient CPU availability and signal $\busy_i$ whenever $\mathit{cpu}_i$ is busy.

    Suppose we obtain a new plant that does not signal the \emph{busy} status in the same way, e.g., a plant that only signals $\busy$ once the CPUs are overloaded.
    For instance, consider the plant $\strat_\plant'$ as in \cref{fig:plant:new}.
    If we use the procedure $\compose(\withprophecies, \strat_\plant')$ to synthesize a controller for this plant, we obtain an explicit controller that is not correct for this plant, as the prophecies are based on the assumption that the plant signals $\busy$ whenever a CPU is busy.
    Hence, we need to refine the approximation using this plant and re-learn the prophecy controller by calling $\learnApprox(\approximation, \{\strat_\plant'\})$.
    
    \cref{fig:refined:controller} shows the refined prophecy controller after learning the new plant.
    Here, at states $q_0$ and $q_1$, the prophecy annotations for outputs $\assign_1$ and $\assign_2$ are represented by the CTL formulas $\Phi_i \coloneqq \forall\nextt(\overload\Rightarrow\assign_i)$, respectively, stating that if there is an overload in the next step, then the controller must have assigned the task to $\mathit{cpu}_i$.
    The controller can only assign a task to a CPU if it is guaranteed that the overload in the next step can only be caused by the other CPU.
    Therefore, the refined prophecy controller is correct for the new plant $\strat_\plant'$, and can be used to synthesize controllers for plants that implement similar behavior.
\end{example}

\begin{figure*}[!ht]
    \centering
    \includegraphics[width=0.3\textwidth]{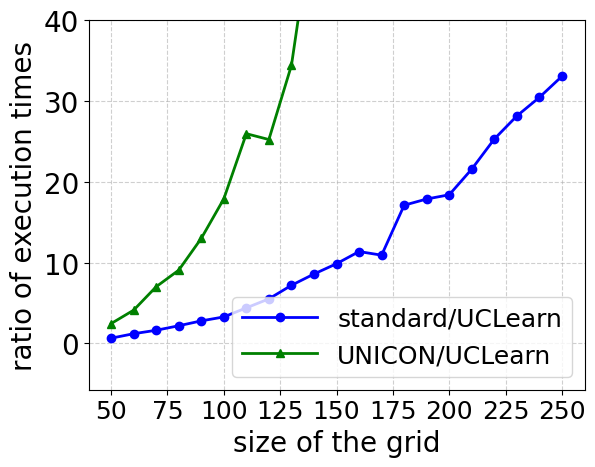}
    \hfill
    \includegraphics[width=0.3\textwidth]{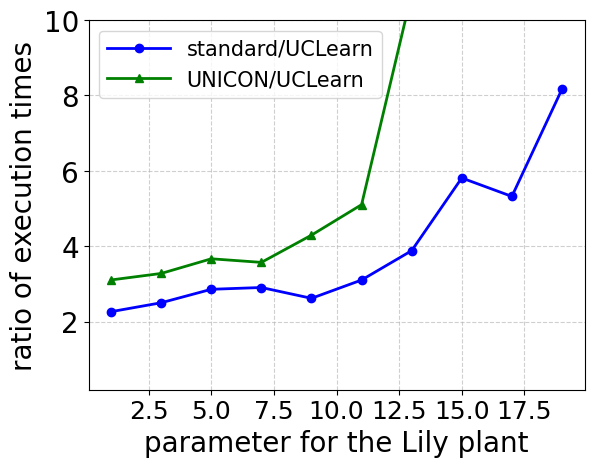}
    \hfill
    \includegraphics[width=0.305\textwidth]{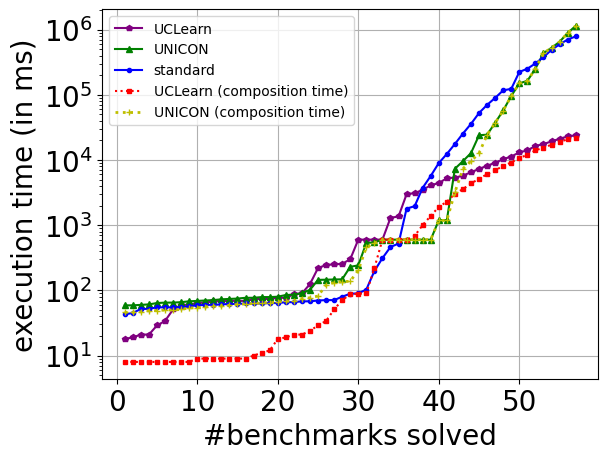}
    \caption{Experimental results showing scalability on the grid world (left) and lily (middle), and overall comparison (right). 
    }
    \label{fig:experiments}
\end{figure*}
\section{Experimental Evaluation}
\label{sec:experiments}
We implemented the algorithms presented in~\Cref{sec:learning,sec:approximations} in an F\#-based prototype tool called \prototype{}~\cite{uclearn} to assess the effectiveness of our method.
The prototype utilizes \textsc{spot}~\cite{spot} for LTL and automata operations, \textsc{oink}~\cite{DBLP:conf/tacas/Dijk18} for game solving, and \cite{learnCTL} for CTL learning.
Additionally, we use \unicon{}~\cite{tacaspaper} to compare our approach with the USC synthesis algorithm and the standard reactive synthesis algorithm.

As a baseline, we tested the running example from \cref{sec:introduction:example,ex:full} using \prototype{}, and the resulting prophecy controller with CTL prophecies are shown in \cref{fig:approx:controller,fig:refined:controller}.
This example illustrates that the learned prophecies are both simple and interpretable.

We further evaluated \prototype{} using a diverse set of benchmarks, including standard reactive synthesis benchmarks from SYNTCOMP~\cite{DBLP:journals/corr/abs-2206-00251} and a set of benchmarks based on a robot grid world, which are detailed below. All experiments were performed on an Apple M1 Pro 8-core CPU and 16GB of RAM. The presented execution times are averages over three runs.

\smallskip
\noindent\textbf{Scalability in Grid World.}
In this benchmark, the goal is to control a robot navigating an $n \times n$ grid. The plant encodes the grid structures with obstacles and the robot's position.
The robot can move in four directions (up, down, left, right), and the controller's task is to ensure that the robot can move in the grid while avoiding obstacles, encoded as a simple LTL formula: $\globally \neg \mathit{collision}$. We evaluate the performance of each approach on increasing grid sizes, and show the ratio of execution times between the existing approaches and \prototype{} in \cref{fig:experiments} (left). The results indicate that our approach is much more scalable and over an order of magnitude faster than the other approaches. This is not surprising as \unicon{} requires capturing all realizations of the plant, while \prototype{} only needs to capture the realizations that are similar to a given set of plants. Furthermore, it is worth noting that the approximation was learned from a single plant with grid size $2$ and consists of CTL formulas of size at most $2$. This demonstrates that our approach is capable of learning a suitable approximation from a small benchmark and that the resulting formulas are human-readable.

\smallskip
\noindent\textbf{Scalability in Lily.}
This benchmark, taken from SYNTCOMP~\cite{DBLP:journals/corr/abs-2206-00251}, requires the controller to grant or cancel a request from a user within 3 time steps and ensure that no two requests are granted consecutively. This is encoded as the LTL formula: $\globally (\mathit{request} \rightarrow \nextt (\mathit{grant} \lor \mathit{cancel} \lor \nextt (\mathit{grant} \lor \mathit{cancel} \lor \nextt (\mathit{grant} \lor \mathit{cancel})))) \land \globally (\mathit{grant} \rightarrow \nextt \neg \mathit{grant})$.
We evaluate all approaches on this benchmark with plants encoding increasing complexity of requests, and the results are shown in \cref{fig:experiments} (middle). As with the grid world benchmark, the approximation is learned from a single plant with parameter $2$, and the learning approach is up to 8 times faster than the standard approach.

\smallskip
\noindent\textbf{Adaptability in SYNTCOMP.}
In addition to demonstrating the scalability of our approach, we also evaluate its adaptability to different plants. We use the above benchmarks and a set of safety benchmarks from SYNTCOMP, each consisting of an assumption and a guarantee.
For such benchmarks, we obtain a plant satisfying the assumption and learn an approximation for the guarantee.
This learned approximation is then used to obtain a controller for a plant of similar size satisfying the assumption.
The results of this evaluation are shown in \cref{fig:experiments} (right).
The plot also compares the time required to compose the (already learned/constructed) prophecy controller with the plant, labeled as the composition time, for both \prototype{} and \unicon{}. 
These results indicate that \prototype{} adapts to changes in the plant model much faster than both the standard approach and \unicon{}.
This is due to the fact that the CTL formulas learned from the plant are concise and small (with a size of at most $4$), making them easily reusable for similar plants.

\section{Conclusion}

In this paper, we introduced a new method for universal controller synthesis, where the controllers are initially created independently of specific plants and are then refined for a given set of nominal plants in a subsequent step.
Universal controllers offer multiple advantages over standard controller synthesis approaches, including scalability, adaptability, and explainability.
We have proposed approximate universal controllers, which use prophecies tailored to a particular class of plants.
The prophecies are represented as CTL formulas and learned from positive and negative examples gathered during explicit controller construction.
Our experiments demonstrate that our approach outperforms existing universal controller synthesis methods whenever the learned prophecies are adequate for building an explicit solution.

\section*{Acknowledgments}
This work is funded by the DFG grant 389792660 as part of TRR 248 – CPEC, by the European Union with ERC
Grant HYPER (No. 101055412), and by the Emmy Noether Grant SCHM 3541/1-1.
Views and opinions expressed are however those of the author(s) only and do not necessarily reflect those of the European Union or the European Research Council Executive Agency. Neither the European Union nor the granting authority can be held responsible for them.

\bibliography{bib}

@inproceedings{spot,
  author    = {Alexandre Duret{-}Lutz and
               Etienne Renault and
               Maximilien Colange and
               Florian Renkin and
               Alexandre Gbaguidi Aisse and
               Philipp Schlehuber{-}Caissier and
               Thomas Medioni and
               Antoine Martin and
               J{\'{e}}r{\^{o}}me Dubois and
               Cl{\'{e}}ment Gillard and
               Henrich Lauko},
  title     = {From Spot 2.0 to Spot 2.10: What's New?},
  booktitle = {International Conference on Computer Aided Verification, {CAV} 2022},
  series    = {Lecture Notes in Computer Science},
  volume    = {13372},
  publisher = {Springer},
  year      = {2022},
  doi       = {10.1007/978-3-031-13188-2\_9}
}

@software{uclearn,
  author       = {Nayak, Satya Prakash and
                  Metzger, Niklas and
                  Schmuck, Anne-Kathrin and
                  Finkbeiner, Bernd},
  title        = {Artifact for "Universal Safety Controllers with
                   Learned Prophecies" at AAAI 2026
                  },
  month        = nov,
  year         = 2025,
  publisher    = {Zenodo},
  doi          = {10.5281/zenodo.17610623},
  url          = {https://doi.org/10.5281/zenodo.17610623},
  swhid        = {swh:1:dir:c06e5423497c4528b0d82810305aae5d6b500333
                   ;origin=https://doi.org/10.5281/zenodo.17610622;vi
                   sit=swh:1:snp:3df459a23839376d49410f6daf50f80fe21f
                   133f;anchor=swh:1:rel:1a865fd51a62919c0ab340ed79ca
                   bc7bcbcc7961;path=/
                  },
}

@inproceedings{DBLP:conf/tacas/Dijk18,
  author       = {Tom van Dijk},
  editor       = {Dirk Beyer and
                  Marieke Huisman},
  title        = {Oink: An Implementation and Evaluation of Modern Parity Game Solvers},
  booktitle    = {Tools and Algorithms for the Construction and Analysis of Systems
                  - 24th International Conference, {TACAS} 2018, Held as Part of the
                  European Joint Conferences on Theory and Practice of Software, {ETAPS}
                  2018, Thessaloniki, Greece, April 14-20, 2018, Proceedings, Part {I}},
  series       = {Lecture Notes in Computer Science},
  volume       = {10805},
  pages        = {291--308},
  publisher    = {Springer},
  year         = {2018},
  url          = {https://doi.org/10.1007/978-3-319-89960-2\_16},
  doi          = {10.1007/978-3-319-89960-2\_16},
  bibsource    = {dblp computer science bibliography, https://dblp.org}
}

@inproceedings{learnCTL,
  author       = {Benjamin Bordais and
                  Daniel Neider and
                  Rajarshi Roy},
  title        = {Learning Branching-Time Properties in {CTL} and {ATL} via Constraint
                  Solving},
  booktitle    = {{FM} {(1)}},
  series       = {Lecture Notes in Computer Science},
  volume       = {14933},
  pages        = {304--323},
  publisher    = {Springer},
  year         = {2024}
}

@book{baier2008principles,
	title={Principles of model checking},
	author={Baier, Christel and Katoen, Joost-Pieter},
	year={2008},
	publisher={MIT press}
}

@inproceedings{safeLtl,
  author       = {Timo Latvala},
  editor       = {Thomas Ball and
                  Sriram K. Rajamani},
  title        = {Efficient Model Checking of Safety Properties},
  booktitle    = {Model Checking Software, 10th International {SPIN} Workshop. Portland,
                  OR, USA, May 9-10, 2003, Proceedings},
  series       = {Lecture Notes in Computer Science},
  volume       = {2648},
  pages        = {74--88},
  publisher    = {Springer},
  year         = {2003},
  url          = {https://doi.org/10.1007/3-540-44829-2\_5},
  doi          = {10.1007/3-540-44829-2\_5},
  bibsource    = {dblp computer science bibliography, https://dblp.org}
}

@inproceedings{AnandNS23,
  author       = {Ashwani Anand and
                  Satya Prakash Nayak and
                  Anne{-}Kathrin Schmuck},
  editor       = {Constantin Enea and
                  Akash Lal},
  title        = {Synthesizing Permissive Winning Strategy Templates for Parity Games},
  booktitle    = {Computer Aided Verification - 35th International Conference, {CAV}
                  2023, Paris, France, July 17-22, 2023, Proceedings, Part {I}},
  series       = {Lecture Notes in Computer Science},
  volume       = {13964},
  pages        = {436--458},
  publisher    = {Springer},
  year         = {2023},
  url          = {https://doi.org/10.1007/978-3-031-37706-8\_22},
  doi          = {10.1007/978-3-031-37706-8\_22},
  bibsource    = {dblp computer science bibliography, https://dblp.org}
}

@proceedings{automata-book,
  editor       = {Erich Gr{\"{a}}del and
                  Wolfgang Thomas and
                  Thomas Wilke},
  title        = {Automata, Logics, and Infinite Games: {A} Guide to Current Research
                  [outcome of a Dagstuhl seminar, February 2001]},
  series       = {Lecture Notes in Computer Science},
  volume       = {2500},
  publisher    = {Springer},
  year         = {2002},
  url          = {https://doi.org/10.1007/3-540-36387-4},
  doi          = {10.1007/3-540-36387-4},
  isbn         = {3-540-00388-6},
  bibsource    = {dblp computer science bibliography, https://dblp.org}
}

@article{bernet2002:permissiveStrategies,
	author    = {Julien Bernet and
	               David Janin and
	               Igor Walukiewicz},
	  title     = {Permissive strategies: from parity games to safety games},
	  journal   = {{RAIRO} Theor. Informatics Appl.},
	  volume    = {36},
	  number    = {3},
	  pages     = {261--275},
	  year      = {2002},
	  
	  doi       = {10.1051/ita:2002013},
	  bibsource = {dblp computer science bibliography, https://dblp.org}
}

@InProceedings{bouyer2011:measuringpermissiveness,
	author    = {Patricia Bouyer and
	               Nicolas Markey and
	               J{\"{o}}rg Olschewski and
	               Michael Ummels},
	  editor    = {Tevfik Bultan and
	               Pao{-}Ann Hsiung},
	  title     = {Measuring Permissiveness in Parity Games: Mean-Payoff Parity Games
	               Revisited},
	  booktitle = {Automated Technology for Verification and Analysis, 9th International
	               Symposium, {ATVA} 2011, Taipei, Taiwan, October 11-14, 2011. Proceedings},
	  series    = {Lecture Notes in Computer Science},
	  volume    = {6996},
	  pages     = {135--149},
	  publisher = {Springer},
	  year      = {2011},
	  
	  doi       = {10.1007/978-3-642-24372-1\_11},
	  bibsource = {dblp computer science bibliography, https://dblp.org}
	}

@inproceedings{DammF14,
  author       = {Werner Damm and
                  Bernd Finkbeiner},
  editor       = {Cliff B. Jones and
                  Pekka Pihlajasaari and
                  Jun Sun},
  title        = {Automatic Compositional Synthesis of Distributed Systems},
  booktitle    = {{FM} 2014: Formal Methods - 19th International Symposium, Singapore,
                  May 12-16, 2014. Proceedings},
  series       = {Lecture Notes in Computer Science},
  volume       = {8442},
  pages        = {179--193},
  publisher    = {Springer},
  year         = {2014},
  url          = {https://doi.org/10.1007/978-3-319-06410-9\_13},
  doi          = {10.1007/978-3-319-06410-9\_13},
  bibsource    = {dblp computer science bibliography, https://dblp.org}
}

@inproceedings{Berwanger07,
  author       = {Dietmar Berwanger},
  editor       = {Wolfgang Thomas and
                  Pascal Weil},
  title        = {Admissibility in Infinite Games},
  booktitle    = {{STACS} 2007, 24th Annual Symposium on Theoretical Aspects of Computer
                  Science, Aachen, Germany, February 22-24, 2007, Proceedings},
  series       = {Lecture Notes in Computer Science},
  volume       = {4393},
  pages        = {188--199},
  publisher    = {Springer},
  year         = {2007},
  url          = {https://doi.org/10.1007/978-3-540-70918-3\_17},
  doi          = {10.1007/978-3-540-70918-3\_17},
  bibsource    = {dblp computer science bibliography, https://dblp.org}
}

@inproceedings{AnandMNS23,
  author       = {Ashwani Anand and
                  Kaushik Mallik and
                  Satya Prakash Nayak and
                  Anne{-}Kathrin Schmuck},
  editor       = {Sriram Sankaranarayanan and
                  Natasha Sharygina},
  title        = {Computing Adequately Permissive Assumptions for Synthesis},
  booktitle    = {Tools and Algorithms for the Construction and Analysis of Systems
                  - 29th International Conference, {TACAS} 2023, Held as Part of the
                  European Joint Conferences on Theory and Practice of Software, {ETAPS}
                  2022, Paris, France, April 22-27, 2023, Proceedings, Part {II}},
  series       = {Lecture Notes in Computer Science},
  volume       = {13994},
  pages        = {211--228},
  publisher    = {Springer},
  year         = {2023},
  url          = {https://doi.org/10.1007/978-3-031-30820-8\_15},
  doi          = {10.1007/978-3-031-30820-8\_15},
  bibsource    = {dblp computer science bibliography, https://dblp.org}
}

@inproceedings{BeutnerF22,
  author       = {Raven Beutner and
                  Bernd Finkbeiner},
  title        = {Prophecy Variables for Hyperproperty Verification},
  booktitle    = {35th {IEEE} Computer Security Foundations Symposium, {CSF} 2022, Haifa,
                  Israel, August 7-10, 2022},
  pages        = {471--485},
  publisher    = {{IEEE}},
  year         = {2022},
  url          = {https://doi.org/10.1109/CSF54842.2022.9919658},
  doi          = {10.1109/CSF54842.2022.9919658},
  bibsource    = {dblp computer science bibliography, https://dblp.org}
}

@article{AbadiL91,
  author       = {Mart{\'{\i}}n Abadi and
                  Leslie Lamport},
  title        = {The Existence of Refinement Mappings},
  journal      = {Theor. Comput. Sci.},
  volume       = {82},
  number       = {2},
  pages        = {253--284},
  year         = {1991},
  url          = {https://doi.org/10.1016/0304-3975(91)90224-P},
  doi          = {10.1016/0304-3975(91)90224-P},
  bibsource    = {dblp computer science bibliography, https://dblp.org}
}

@inproceedings{Pnueli77,
  author    = {Amir Pnueli},
  title     = {The Temporal Logic of Programs},
  booktitle = {18th Annual Symposium on Foundations of Computer Science, Providence,
               Rhode Island, USA, 31 October - 1 November 1977},
  pages     = {46--57},
  publisher = {{IEEE} Computer Society},
  year      = {1977},
  url       = {https://doi.org/10.1109/SFCS.1977.32},
  doi       = {10.1109/SFCS.1977.32},
  bibsource = {dblp computer science bibliography, https://dblp.org}
}

@Article{Klein2015:mostGeneralController,
	author    = {Joachim Klein and
	               Christel Baier and
	               Sascha Kl{\"{u}}ppelholz},
	  title     = {Compositional construction of most general controllers},
	  journal   = {Acta Informatica},
	  volume    = {52},
	  number    = {4-5},
	  pages     = {443--482},
	  year      = {2015},
	  
	  doi       = {10.1007/s00236-015-0239-9},
	  bibsource = {dblp computer science bibliography, https://dblp.org}
}

@inproceedings{FinkbeinerP22,
  author       = {Bernd Finkbeiner and
                  Noemi Passing},
  editor       = {Anuj Dawar and
                  Venkatesan Guruswami},
  title        = {Synthesizing Dominant Strategies for Liveness},
  booktitle    = {42nd {IARCS} Annual Conference on Foundations of Software Technology
                  and Theoretical Computer Science, {FSTTCS} 2022, December 18-20, 2022,
                  {IIT} Madras, Chennai, India},
  series       = {LIPIcs},
  volume       = {250},
  pages        = {37:1--37:19},
  publisher    = {Schloss Dagstuhl - Leibniz-Zentrum f{\"{u}}r Informatik},
  year         = {2022},
  url          = {https://doi.org/10.4230/LIPIcs.FSTTCS.2022.37},
  doi          = {10.4230/LIPICS.FSTTCS.2022.37},
  bibsource    = {dblp computer science bibliography, https://dblp.org}
}

@article{BrenguierRS17,
  author       = {Romain Brenguier and
                  Jean{-}Fran{\c{c}}ois Raskin and
                  Ocan Sankur},
  title        = {Assume-admissible synthesis},
  journal      = {Acta Informatica},
  volume       = {54},
  number       = {1},
  pages        = {41--83},
  year         = {2017},
  url          = {https://doi.org/10.1007/s00236-016-0273-2},
  doi          = {10.1007/S00236-016-0273-2},
  bibsource    = {dblp computer science bibliography, https://dblp.org}
}

@inproceedings{BassetRS17,
  author       = {Nicolas Basset and
                  Jean{-}Fran{\c{c}}ois Raskin and
                  Ocan Sankur},
  editor       = {Luca Aceto and
                  Giorgio Bacci and
                  Giovanni Bacci and
                  Anna Ing{\'{o}}lfsd{\'{o}}ttir and
                  Axel Legay and
                  Radu Mardare},
  title        = {Admissible Strategies in Timed Games},
  booktitle    = {Models, Algorithms, Logics and Tools - Essays Dedicated to Kim Guldstrand
                  Larsen on the Occasion of His 60th Birthday},
  series       = {Lecture Notes in Computer Science},
  volume       = {10460},
  pages        = {403--425},
  publisher    = {Springer},
  year         = {2017},
  url          = {https://doi.org/10.1007/978-3-319-63121-9\_20},
  doi          = {10.1007/978-3-319-63121-9\_20},
  bibsource    = {dblp computer science bibliography, https://dblp.org}
}

@inproceedings{neider2018learning,
  title={Learning linear temporal properties},
  author={Neider, Daniel and Gavran, Ivan},
  booktitle={2018 Formal Methods in Computer Aided Design (FMCAD)},
  pages={1--10},
  year={2018},
  organization={IEEE}
}

@inproceedings{BalachanderFR23,
  author       = {Mrudula Balachander and
                  Emmanuel Filiot and
                  Jean{-}Fran{\c{c}}ois Raskin},
  title        = {{LTL} Reactive Synthesis with a Few Hints},
  booktitle    = {{TACAS} {(2)}},
  series       = {Lecture Notes in Computer Science},
  volume       = {13994},
  pages        = {309--328},
  publisher    = {Springer},
  year         = {2023}
}

@inproceedings{KretinskyMPZ25,
  author       = {Jan Kret{\'{\i}}nsk{\'{y}} and
                  Tobias Meggendorfer and
                  Maximilian Prokop and
                  Ashkan Zarkhah},
  title        = {SemML: Enhancing Automata-Theoretic {LTL} Synthesis with Machine Learning},
  booktitle    = {{TACAS} {(1)}},
  series       = {Lecture Notes in Computer Science},
  volume       = {15696},
  pages        = {233--253},
  publisher    = {Springer},
  year         = {2025}
}

@inproceedings{pommellet2024sat,
  title={Sat-based learning of computation tree logic},
  author={Pommellet, Adrien and Stan, Daniel and Scatton, Simon},
  booktitle={International Joint Conference on Automated Reasoning},
  pages={366--385},
  year={2024},
  organization={Springer}
}

@inproceedings{valizadeh2024ltl,
  title={Ltl learning on gpus},
  author={Valizadeh, Mojtaba and Fijalkow, Nathana{\"e}l and Berger, Martin},
  booktitle={International Conference on Computer Aided Verification},
  pages={209--231},
  year={2024},
  organization={Springer}
}

@inproceedings{raha2023synthesizing,
  title={Synthesizing efficiently monitorable formulas in metric temporal logic},
  author={Raha, Ritam and Roy, Rajarshi and Fijalkow, Nathana{\"e}l and Neider, Daniel and P{\'e}rez, Guillermo A},
  booktitle={International Conference on Verification, Model Checking, and Abstract Interpretation},
  pages={264--288},
  year={2023},
  organization={Springer}
}

@article{DBLP:journals/corr/abs-2206-00251,
  author       = {Swen Jacobs and
                  Guillermo A. P{\'{e}}rez and
                  Remco Abraham and
                  V{\'{e}}ronique Bruy{\`{e}}re and
                  Micha{\"{e}}l Cadilhac and
                  Maximilien Colange and
                  Charly Delfosse and
                  Tom van Dijk and
                  Alexandre Duret{-}Lutz and
                  Peter Faymonville and
                  Bernd Finkbeiner and
                  Ayrat Khalimov and
                  Felix Klein and
                  Michael Luttenberger and
                  Klara J. Meyer and
                  Thibaud Michaud and
                  Adrien Pommellet and
                  Florian Renkin and
                  Philipp Schlehuber{-}Caissier and
                  Mouhammad Sakr and
                  Salomon Sickert and
                  Ga{\"{e}}tan Staquet and
                  Cl{\'{e}}ment Tamines and
                  Leander Tentrup and
                  Adam Walker},
  title        = {The Reactive Synthesis Competition {(SYNTCOMP):} 2018-2021},
  journal      = {CoRR},
  volume       = {abs/2206.00251},
  year         = {2022},
  url          = {https://doi.org/10.48550/arXiv.2206.00251},
  doi          = {10.48550/ARXIV.2206.00251},
  eprinttype    = {arXiv},
  eprint       = {2206.00251},
  bibsource    = {dblp computer science bibliography, https://dblp.org}
}

@InProceedings{FremontSeshi_Improvisation,
author="Fremont, Daniel J.
and Seshia, Sanjit A.",
editor="Chockler, Hana
and Weissenbacher, Georg",
title="Reactive Control Improvisation",
booktitle="Computer Aided Verification",
year="2018",
publisher="Springer International Publishing",
address="Cham",
pages="307--326",
isbn="978-3-319-96145-3"
}

@book{tabuada2009verification,
  title={Verification and control of hybrid systems: a symbolic approach},
  author={Tabuada, Paulo},
  year={2009},
  publisher={Springer Science \& Business Media}
}

@article{Review24_FMandControl_YinGaoYu,
title = {Formal synthesis of controllers for safety-critical autonomous systems: Developments and challenges},
journal = {Annual Reviews in Control},
volume = {57},
pages = {100940},
year = {2024},
issn = {1367-5788},
doi = {https://doi.org/10.1016/j.arcontrol.2024.100940},
url = {https://www.sciencedirect.com/science/article/pii/S1367578824000099},
author = {Xiang Yin and Bingzhao Gao and Xiao Yu},
}

@book{belta2017formal,
  title={Formal methods for discrete-time dynamical systems},
  author={Belta, Calin and Yordanov, Boyan and Gol, Ebru Aydin},
  volume={89},
  year={2017},
  publisher={Springer}
}

@book{cassandras2021introduction,
  title={Introduction to Discrete Event Systems},
  author={Cassandras, Christos G and Lafortune, St{\'e}phane},
  year={2021},
  publisher={Springer Nature}
}

@inproceedings{DBLP:conf/lop/ClarkeE81,
  author       = {Edmund M. Clarke and
                  E. Allen Emerson},
  editor       = {Dexter Kozen},
  title        = {Design and Synthesis of Synchronization Skeletons Using Branching-Time
                  Temporal Logic},
  booktitle    = {Logics of Programs, Workshop, Yorktown Heights, New York, USA, May
                  1981},
  series       = {Lecture Notes in Computer Science},
  volume       = {131},
  pages        = {52--71},
  publisher    = {Springer},
  year         = {1981},
  url          = {https://doi.org/10.1007/BFb0025774},
  doi          = {10.1007/BFB0025774},
  bibsource    = {dblp computer science bibliography, https://dblp.org}
}

@article{DBLP:journals/toplas/ClarkeES86,
  author       = {Edmund M. Clarke and
                  E. Allen Emerson and
                  A. Prasad Sistla},
  title        = {Automatic Verification of Finite-State Concurrent Systems Using Temporal
                  Logic Specifications},
  journal      = {{ACM} Trans. Program. Lang. Syst.},
  volume       = {8},
  number       = {2},
  pages        = {244--263},
  year         = {1986},
  url          = {https://doi.org/10.1145/5397.5399},
  doi          = {10.1145/5397.5399},
  bibsource    = {dblp computer science bibliography, https://dblp.org}
}

@inproceedings{tacaspaper,
  author       = {Bernd Finkbeiner and
                  Niklas Metzger and
                  Satya Prakash Nayak and
                  Anne{-}Kathrin Schmuck},
  editor       = {Arie Gurfinkel and
                  Marijn Heule},
  title        = {Synthesis of Universal Safety Controllers},
  booktitle    = {Tools and Algorithms for the Construction and Analysis of Systems
                  - 31st International Conference, {TACAS} 2025, Held as Part of the
                  International Joint Conferences on Theory and Practice of Software,
                  {ETAPS} 2025, Hamilton, ON, Canada, May 3-8, 2025, Proceedings, Part
                  {II}},
  series       = {Lecture Notes in Computer Science},
  volume       = {15697},
  pages        = {177--197},
  publisher    = {Springer},
  year         = {2025},
  url          = {https://doi.org/10.1007/978-3-031-90653-4\_9},
  doi          = {10.1007/978-3-031-90653-4\_9},
  bibsource    = {dblp computer science bibliography, https://dblp.org}
}

\end{document}